\documentclass{bmcart}
\usepackage{graphicx} 

\usepackage{amsmath,amsfonts,amsbsy,amssymb} 
\usepackage{mathabx} 
\usepackage{amssymb} 
\usepackage{amsmath}
\usepackage{amsthm}	
\usepackage{mathrsfs}
\usepackage[nolist]{acronym} 
\usepackage{tabularx} 
\usepackage{multirow}
\usepackage{wasysym}
\usepackage{float}

\usepackage{color} 

\usepackage{enumitem} 

\newtheorem{proposition}{Proposition}
\newtheorem{lemma}{Lemma}

\newtheorem*{remark}{Remark}

\usepackage[utf8]{inputenc} 

\startlocaldefs
\endlocaldefs

\begin{document}

\begin{frontmatter}

\begin{fmbox}
\dochead{Research}


\title{Performance Analysis of Ultra-Reliable Short Message Decode and Forward Relaying Protocols}


\author[
addressref={aff1},                   
corref={aff1},                       
email={Parisa.Nouri@oulu.fi}   
]{\inits{PN}\fnm{Parisa} \snm{Nouri}}
\author[
addressref={aff1},
email={Hirley.Alves@oulu.fi}
]{\inits{HA}\fnm{Hirley} \snm{Alves}}
\author[
addressref={aff1},
email={Matti.Latva-aho@oulu.fi}
]{\inits{ML}\fnm{Matti} \snm{Latva-aho}}


\address[id=aff1]{
	\orgname{Centre for Wireless Communications (CWC), University of Oulu, Finland} 
}



\end{fmbox}


\begin{abstractbox}
	
	\begin{abstract} 
	Machine-Type Communication (MTC) is a rapidly growing technology which covers a broad range of automated applications and propels the world into a fully connected society. Two new use cases of MTC are mMTC and URLLC, where mMTC support a large number of devices with high reliability and low rate connectivity while URLLC refers to excessively low outage probability under very stringent latency constraint. Herein, we examine the URLLC through three cooperative schemes, namely dual-hop DF, SC and MRC, and compare to direct transmission under Rayleigh fading. We compare the performance of studied cooperative protocols under two distinct power constraints with respect to latency and energy efficiency. Moreover, we illustrate the impact of coding rate on the probability of successful transmission in ultra-reliable region in addition to the effect of power allocation on the outage probability. We also provide the performance analysis of cooperative schemes in terms of energy efficiency and latency requirements. 
	\end{abstract}




\end{abstractbox}
%

\end{frontmatter}



\section{Introduction}
The fifth generation (5G) of cellular networks for beyond 2020 envisages to handle two new use cases in Machine-Type Communications (MTC), namely Ultra-Reliable Low Latency Communications (URLLC) and massive MTC (mMTC)~\cite{popovski2017ultra}. In MTC, MTC devices autonomously communicate with minimum human cooperation~\cite{shariatmadari2015machine},~\cite{mehmood2015mobile}. 5G communication technology should be flexible enough to support ultra-reliable low latency communications by guaranteeing reliability greater than $99.999 \%$~\cite{kalor2017network}. Key challenges and requirements of 5G technology such as latency, data rate, energy and cost issues are discussed in more details in~\cite{popovski2017ultra},~\cite{andrews2014will},~\cite{dawy2017toward}.

In recent years, MTC has gained much attention from the mobile network operators, equipment vendors and academic researchers due to such novel communication paradigm, the capability of exchanging short data messages and also being cost-effective, energy efficient\cite{lee2016packet},~\cite{taleb2012machine}; reliable and within a stringent delay requirement. MTC takes advantage of several distinctive properties such as group-based communications, low mobility, time-controlled, time-tolerant and secure connection which are at the same time challenging tasks since technically advanced
solutions are needed to deliver the required tasks. Within the application requirements, hence, opening up different research areas is currently being carried out in academia, industry, and standards bodies~\cite{condoluci2015toward}. Current technologies cover a small range of applications and services while the upcoming MTC should be able to cover a broad range of services with multiple forms of data traffic in order to deal with different service requirements as data rate, latency, reliability, energy consumption and security~\cite{dawy2017toward},~\cite{monserrat2015metis}. Future MTC improvements will be conspicuous in health-care, logistics, process automation, transportation, e.g~\cite{bockelmann2016massive},~\cite{dawy2017toward}. 
In mMTC a huge number of devices in a specific domain are connected to the cellular network with low-rate and low-power connectivity, different quality-of-service (QOS) requirements and high reliability to support demanding situations, e.g. smart meters, actuators~\cite{popovski2014ultra},~\cite{singh2016selective}. 

Moreover, MTC services have to met stringent timing constraints from few seconds to even excessively low end-to-end deadlines in mission critical communications~\cite{biral2016impact}, connection between vehicles, remote control of robots in addition to an extreme low end-to-end latency in the scope of less than a millisecond which is a key enabler in several services including cloud connectivity, industrial control, road safety~\cite{popovski2017ultra},~\cite{bockelmann2016massive},~\cite{durisi2016toward}. 
Latency refers to the time duration between transferring the message from the transmitter and receiving correctly at the receiver where some messages drop due to the buffer overflows, unsuccessful synchronizations, unsuccessful decoding which result in unlimited delay~\cite{popovski2017ultra}. Hence, we can define the reliability as the probability of successful transmission under the predetermined delay constraint~\cite{popovski2017ultra},~\cite{yilmaz2016ultra}. In URLLC, high probability of successful transmission indicates low outage probability (or packet drop) while the opposite does not always hold as the reliability is restricted to a specific latency budget due to the limited amount of channel uses~\cite{popovski2017ultra}. 
Hence, one of the major requirements of URLLC is a extremely low outage probability under a very demanding latency budget where  retransmissions are not always available. In the use of short messages under URLLC, new robust channel codes are needed; otherwise, the performance of the system will be even further away from the Shannon limit with long data packets~\cite{sybis2016channel}. 

Under Shannon's channel coding theorem, error-free communication is attained when the blocklength goes to infinity~\cite{hu2015performance}. For instance, authors in~\cite{polyanskiy2010channel}, provide a tight approximation of achievable coding rate under finite blocklength (FB) regime and indicate a noticeable performance loss compared to the Shannon coding. This motivates us to analyze the performance of MTC under FB regime
since in URLLC, due to the equal packet length of metadata and information bits, an unsuccessful encoding of the metadata decreases the system efficiency~\cite{durisi2016toward}. In the past few years, several works have studied different aspects of FB coding since majority of the theoretical results assume infinite blocklength (IFB). For instance, authors in \cite{iscan2016comparison}, examine some possible FB coding schemes which may be applied in 5G technology. They show that novel coding schemes with better minimum distance between the codewords, improve the efficiency of system at the cost of more sophisticated decoders. Moreover, the performance of spectrum sharing networks with FB codes are studied in~\cite{makki2015finite}. The blocklength of information bits highly affects the system quality where an optimal power allocation technique improves the system efficiency with short message transmissions. Furthermore, authors in~\cite{park2012new}, propose a new power allocation technique, so-called modified water-filling in order to maximize the lower bound of the coding rate with short packet transmission compared to the common water-filling method. In addition, performance of ARQ protocol in terms of throughput and average latency is studied in~\cite{devassy2014finite}. Authors determine the optimal lengths of the codeword which minimize the latency and maximize the throughput per-user for an specific number of information bits. They illustrate that with optimal codes, the shorter the codeword is, the lower outage probability attains. 
\subsection{On the Impact of Cooperative Diversity}
Cooperative diversity provides the possibility of high data rate; while improving the reliability. In cooperative networks, intermediate nodes transfer the message from the source to destination~\cite{ikki2007performance}. Cooperative technique exploits the spatial diversity gain to reduce the impact of wireless fading from multipath propagation. The major advantage of this technique is that the several independent copies of a signal arrive at the destination without installing collocated antennas at the source or receiver in addition of a better signal quality, better coverage, greater capacity and lower transmit power~\cite{zimmermann2005performance},~\cite{mansourkiaie2015cooperative}. 
The most conventional cooperative scheme is decode-and-forward (DF), where the auxiliary node, namely relay, decodes, encodes and retransmits the message~\cite{zimmermann2005performance}. Cooperative schemes are categorized as fixed, adaptive and feedback schemes~\cite{zimmermann2005performance}. In the fixed protocol, relay always forwards the message to destination while in adaptive protocol, the relay retransmits the message under a predefined threshold rule which enables that to communicate independently or not. In the feedback protocol, if the destination requests, the cooperation takes place~\cite{zimmermann2005performance}.
During the past few years, the efficiency of cooperative networks has been investigated in several system and channel models. Authors in~\cite{swamy2015cooperative}, propose a method that meet the high reliability and latency requirements through taking the advantage of cooperative relaying technique. Moreover, authors in~\cite{mansourkiaie2015cooperative}, provide a comprehensive study regarding the exiting cooperative schemes and analyze the performance of each scheme. 
Relaying performance of quasi-static Rayleigh channels where the channel gains of the direct link and relaying are combined at the destination, is studied in~\cite{hu2015performance}. They indicate that the performance loss increases if the outage probability of the source-to-relay link is higher than the overall outage probability. The efficiency of multi-relay DF scenario under the assumption of perfect channel-state-information (CSI) and partial CSI is provided in~\cite{hu2016relaying}. Authors show that with perfect CSI, the throughput of IFB is smaller than the throughput with FB coding. Authors in~\cite{du2016finite}, examine the throughput of a multi-hop relaying network under FB and IFB regimes with two assumptions: $i)$target overall outage probability $ii)$constant coding rate. They illustrate that there is different but optimal number of hops which maximize the throughput for either FB or IFB assumptions. In addition, they indicate that the FB-throughput is quasi-concave in the overall outage probability and coding rate. Furthermore, authors illustrate that the multi-hop network is less affected by the blocklength under the constant coding rate assumption compared to the target overall outage probability scenario. Moreover, \cite{beaulieu2006closed} studies the performance of DF relay network in dissimilar Rayleigh fading channels. Although authors attain the closed form expression of the outage probability, but they do not consider the impact of finite blocklength coding. Furthermore, authors in \cite{tran2014achievable}, study the achievable coding rate and ergodic capacity of non-orthogonal amply-and-forward (AF) multi relay network subject a total average power constraint (TAPC) and an individual average power constraint (IAPC). They indicate that the ergodic capacity can be attained by an iterative water-filling-based algorithm. In addition, they show that in a multi relay NAF network, the transmit power at the source should be equally allocated in all broadcasting phases to cover the capacity at sufficiently high SNRs.

Moreover, in our previous work~\cite{nouri2017performance}, introduces relaying as means to achieve ultra-reliable. We study the performance of cooperative relaying protocols, supposing Rayleigh fading channels. We show that relaying technique improves the reliability and  how we can meet the ultra-reliable communication requirements. We examine the impact of coded blocklength and number of information bits on the probability of successful transmission. In addition, it is shown that relaying requires less transmit power compared to the direct transmission (DT) to enable ultra-reliable under FB regime. We also provide an approximation to the outage probability in closed form. We extend our work in~\cite{nouri2017performance}, by considering ultra-reliable MTC with incremental relaying technique in~\cite{nouri2017ultra}. We define the overall outage probability in each studied relaying scheme, assuming Nakagami-$m$ fading. We investigate the impact of fading severity and power allocation factor on the outage probability. We also provide the outage probability in closed form. Our works in~\cite{nouri2017ultra} and~\cite{nouri2017performance} show that cooperative diversity os useful to meet URLLC requirements.
\subsection{Energy Efficiency of Cooperative Communication}
Another key characteristic of wireless communications that highly affect the performance of 5G networks, is the energy efficiency (EE) due to the limited energy resources in energy-constraint networks~\cite{qiao2009energy},~\cite{gursoy2009capacity}. EE which has been widely studied recently literatures, is defined as the ratio of successfully transmitted bits to the total consumed energy~\cite{wu2015recent},\cite{she2016energy}. Hence, reducing the amount of energy-per-bit, improves EE at low SNR regime~\cite{gursoy2009capacity}, particularly in wireless networks where the batteries which are not rechargeable or easy to charge, supply the wireless components~\cite{wu2015recent}. The reason which motivates us to study EE in the context of URLLC is that,
URLLC is achieved at the cost of high transmit power~\cite{dosti2017ultra},~\cite{dosti2017ultraa}, but we aim to show that cooperative diversity alleviates these demands.  

In early works, authors in~\cite{she2016energy},  examine EE in tactile Internet under queuing and transmission delays to design an energy efficient resource allocation strategy. They propose an optimal resource allocation strategy where the average total consumed power under stringent latency constraint equals to that with unlimited queuing latency requirement with plenty of transmit antennas. 
Moreover,~\cite{wu2015recent}, provides a comprehensive overview of energy-efficient networks, and determines the trade-off between energy efficiency and spectrum efficiency and their applications in 5G networks.

\subsection{Our Contribution}
In this work, we further study three cooperative protocols, namely DF, selection combining (SC) and MRC. Furthermore, we indicate the superiority of MRC over SC and DF protocols in terms of coding rate and reliability. We also show the optimal value of power allocation at the source and relay in each of studied protocols. Moreover, we examine the minimum latency and energy efficiency in cooperative schemes under two different power allocation constraints. 
\begin{center}
	\begin{figure}[!t] 
		\centering
		\includegraphics[width=0.57\columnwidth,height=2.in]{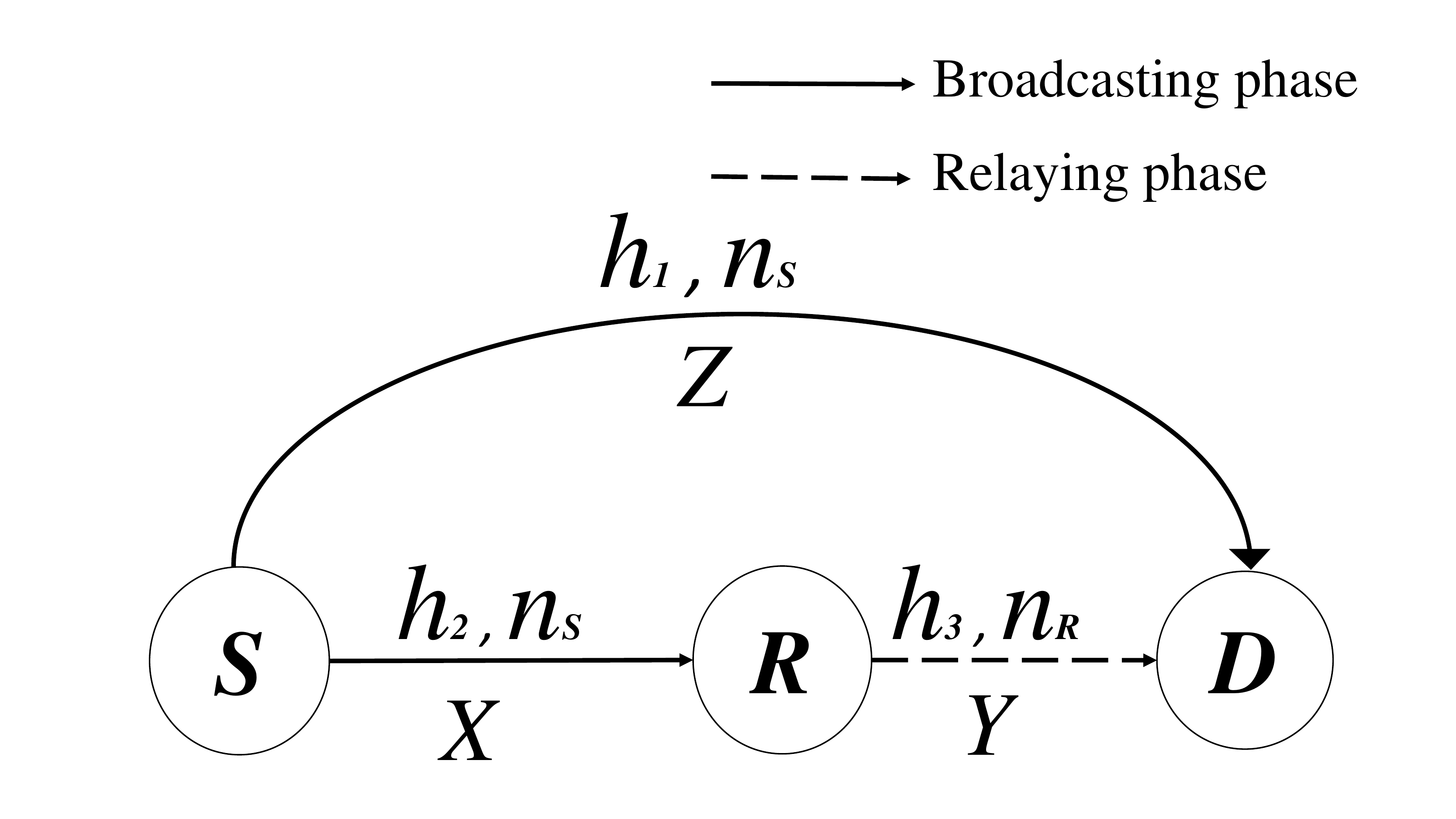}
		\vspace{-0mm}
		\caption{System model for relaying scenario with a source, destination and a decode-forward relay. This figure illustrates the system model for cooperative decode-and-forward relaying scenario including a source, a relay, and a destination. The links between $S$ to $D$, $S$ to $R$, and $R$ to $D$ are referred as the direct link $h_1$, broadcasting link $h_2$ and relaying link $h_3$ each of which with $n_i$ channel uses respectively, where $i\in \{S,R\}$.}
		\label{fig:systemmodel}
		\vspace{-0mm}
	\end{figure}
\end{center}
The following are considered the contributions of this work.
\begin{itemize}
	\item We provide the general expression of the outage probability for each relaying scheme studied in this work.
	\item We extend the work in~\cite{hu2015performance}, by proposing the closed form expression for the outage probability.
	\item We extend our previous works in~\cite{nouri2017performance},~\cite{nouri2017ultra} by studying the minimum latency and energy efficiency under two distinct power constraints, so-called $i$) equal power allocation (EPA) and $ii$) optimal power allocation (OPA) strategies which are allocated numerically. 
	\item We provide the asymptotic analysis of studied cooperative schemes.
\end{itemize}
%

The rest of this paper is organized as follows. Section \ref{systemmodel} presents the system model. Section \ref{sec:Relaying assumptions} discusses the cooperative diversity and examines the outage probability in three cooperative schemes considered in this work, and Section \ref{URLLC} presents some numerical results regarding the performance of studied cooperative schemes under URR. Section \ref{EE} investigates the energy efficiency of considered cooperative schemes and presents some numerical results. Finally, Section \ref{con} concludes the paper.
The important abbreviation and symbols are listed in Table 1.
\begin{table}[h!]\label{T11}
	\caption{SUMMARY OF THE FUNCTIONS AND SYMBOLS.}
	\begin{tabular}{cccc}
		\hline
		bpcu & Bit per Channel Use\\
		CSI & Channel State Information\\
		DF & Decode and Forward\\
		DT & Direct Transmission\\
		EE & Energy Efficiency\\
		EPA & Equal Power Allocation\\
		FB & Finite Blocklength\\
		IFB & Infinite Blocklength\\
		mMTC & Massive Machine-Type Communication\\
		MRC & Maximum Ratio Combining\\
		MTC & Machine-Type Communication\\
		OPA & Optimal Power Allocation\\
		PDF & Probability Density Function\\
		QOS & Quality of Service\\
		RV & Random Variable\\
		SC & Selection Combining \\
		SH & Single Hop\\
		SNR & Signal to Noise Ration\\
		URLLC & Ultra-Reliable Low Latency Communication\\
		URR & Ultra reliable Region \vspace{7mm}\\
		
		$f_W(.)$ & Probability Density Function\\
		$\operatorname{E}(.)$ & Expectation\\
		$\operatorname{Q}(.)$ & $\operatorname{Q}$-Function\\
		$C(\rho)$ & Shannon Capacity \\
		$V(\rho)$ &Channel Dispersion\\
		$\operatorname{Q^{-1}}(.)$ & Inverse of $\operatorname{Q}$-Function \vspace{7mm}\\

		$e$ & Exponential Euler's Number\\
		$y$ & Received Signal\\
		$h$ & Fading Channel\\
		$w$ & AWGN Noise\\
		$k$ & Information Bits\\
		$x_i$ & Transmitted Signal\\
		$n$ & Number of Channel Uses\\
		$\log_2$ & Logarithm to the Base 2\\
	
		$n_S$ & Number of Channel Uses for Source\\
		$n_R$ & Number of Channel Uses for Relay\\
		$d_{SD}$ & Distance of Source-Destination Link\\
		$d_{SR}$ & Distance of Source-Relay Link\\		
		$d_{RD}$ & Distance of Relay-Destination Link\\
		$E$& Energy Efficiency\\
		$P$ & Total Power\\
		$P_R$ & Power of Relay \\
		$P_S$ & Power of Source\\
		$P_{TX}$ & Power of Transmitter \\
		$P_{RX}$ & Power of Receiver\\
		$P_{PA}$ & Power of Amplifier\\
		$P_{succ}$ & Probability of Successful Transmission\\
		$P_{max}$& Maximum Total Power\\
		$N_0$ & Power of Noise\vspace{7mm}\\
		
		$\Omega$ & Instantaneous SNR\\
		$\gamma$ & Average SNR \\
		$\rho$ & Average Power Constraint\\
		$\eta$ & Power Allocation Factor\\
		$\cal R$ & Maximum Coding Rate\\
		$\phi$ & Drain Efficiency\\
		$\epsilon$ & Outage Probability \\\hline

	\end{tabular}
\end{table}

\section{Preliminaries}\label{systemmodel}
\vspace{-0mm}
\subsection{System Model}
\vspace{-0mm}

Fig.$1$ illustrates a DF relaying scenario including a source $S$, destination $D$ and a decode-forward relay $R$. We normalize the distance of $S$ to $D$ as $d_{SD}=1$m, and that $R$ can move in a straight line between $S$ and $D$, while the distance between $S$ and $R$ is denoted by $d_{SR}= \beta d_{SD}$ and the distance of the relaying link is denoted by $d_{RD}= (1-\beta) d_{SD}$. The links denoted by the following random variables $X$, $Y$ and $Z$ represent the $S$-$R$, $R$-$D$ and $S$-$D$ links respectively, and each transmission uses $n_i$ channel uses where $i\in \{S,R\}$. This means that $n_S$ channel uses in the broadcasting phase and $n_R$ channel uses for the relaying phase. In this scenario, first $S$ sends the message to the $D$ and $R$ in the broadcasting phase and  if $R$ successfully decodes the message, forwards it to the $D$ in the relaying phase~\cite{hu2015performance}. The received signals in the broadcasting phase are denoted as $y_1$ and $y_2$, and only if $R$ collaborates with $S$, the received signal at $D$ is $y_3$ as follow~\cite{hu2015performance}
\begin{equation}
y_1 = h_1x+w_1, \;\;\;y_2 = h_2x+w_2 \;\;\;\textrm{and} \;\;\; y_3 = h_3x+w_3,
\end{equation}
where $x$ is the transmitted signal with power $P$ and $w_i$ is the AWGN noise with power $N_0=1$ where $i\in \{X,Y,Z\}$.
Quasi-static Rayleigh fading channels in the $S$-$D$, $S$-$R$ and $R$-$D$ links are denoted as $h_1$, $h_2$ and $h_3$, respectively. In this work, we consider two distinct power constraints, namely $i$) EPA where equal powers are allocated to $S$ and $R$ and $ii$) OPA where total power of $S$ and $R$ is equal to the maximum power.
In a DF-based relaying protocol, the instantaneous SNR depends on the total power constraint $P=P_S+P_R=\eta P+ (1-\eta)P$, which is given by $\Omega_Z = \eta P|h_1|^2 / N_0$, $\Omega_X = \eta P|h_2|^2 / N_0$ and $\Omega_Y = (1-\eta) P|h_3|^2 / N_0$, where $0<\eta\leqslant1$ is the power allocation factor considered to provide a fair comparison between DT and cooperative transmissions and $\eta =0.5$ with EPA. Hence, the average SNR in each link is $\gamma_Z = \eta P/ N_0$, $\gamma_X = \eta P/ N_0$ and $\gamma_Y = (1-\eta) P/ N_0$.

\subsection{Performance Analysis of Single-Hop Communication under the Finite Blocklength Regime}
In this section we revisit the concept of FB coding.  In a single hope communication, first $k$ information bits are mapped to a sequence, namely codeword including $n$ symbols. Afterwards, the created codeword passes the wireless channels and channel outputs map into the estimate of the information bits. Thus, for a single-hop communication with blocklength $n$, outage probability $\epsilon$ and the average power constraint $\rho$, where $\frac{1}{n}\sum_{i}^{n}|x_{i}|^2\leq\rho$ holds, the maximum coding rate  ${\cal R}^*(n,\epsilon)$ of AWGN channel in bits is calculated as
\begin{equation} \label{eq:maximum rate}
{\cal R}^*(n,\epsilon)\!=\! C(\rho) - \sqrt{\frac {V(\rho)}{n}}\operatorname{Q}^{-1}(\epsilon)\log_{2}\operatorname{e},
\end{equation} 
where, $C(\rho)\!=\!\log_{2}(1+\rho)$ is the positive channel capacity and $V(\rho)\!=\!\rho (2+\rho) \big/(1+\rho)^2$ is the channel dispersion~\cite{durisi2016toward}.
According to (\ref{eq:maximum rate}), the outage probability is given by 
\begin{align} \label{eq:outage probability}
\epsilon \!=\! \operatorname{Q}\Bigg(\sqrt{n} \frac{C(\rho)-{\cal R}^*(n,\epsilon)}{\sqrt{V(\rho)}\log_{2}\operatorname{e}}\Bigg),
\end{align}
which holds for the AWGN channels where the channel coefficient $h_{i}$ is equal to one. While for quasi-static fading channels, we attain the outage probability as follow~\cite{durisi2016toward}
\begin{equation} \label{outage_fading}
\epsilon\approx \operatorname{E}\Bigg[\operatorname{Q}\Bigg(\sqrt{n}\frac{C(\rho|h|^2)-{\cal R}^{*}(n,\epsilon) }{\sqrt{V(\rho|h|^2)}}\Bigg)\Bigg].
\end{equation}
Note that (\ref{outage_fading}) is accurate for $n > 100$, as proved for AWGN channels \cite[Figs. $12$ and $13$]{polyanskiy2010channel}, as well as for fading channels as discussed in~\cite{yang2014quasi}.
In addition, in the relaying schemes, we assume that $S$ can encode $k$ information bits into $n_{S}$ channel uses, while $R$ uses $n_{R}$ channel uses. Hence, $S$ and $R$ could employ more sophisticated encoding technique than~\cite{hu2015performance},~\cite{hu2016blocklength}.
\subsection{Closed-Form Expression of the Outage Probability}
Unfortunately, (\ref{outage_fading}) does not have a closed-form expression, but it can be tightly approximated as we shall see next. 
\begin{lemma} 
	The outage probability is approximated as
	\begin{align}\label{outage_fading_linear}
	\epsilon\!=\!1\!-\!\dfrac{\zeta}{\sqrt{2\pi}}\exp(-\theta_m)\bigg[\exp\bigg(\sqrt{\dfrac{\pi}{2\zeta^2}}\bigg)\!-\!\exp\bigg(-\sqrt{\dfrac{\pi}{2\zeta^2}}\bigg)\bigg],
	\end{align}
	where $\theta_m\!=\!\tfrac{2^{\cal R}-1}{P}$ and $\zeta= P\sqrt{2\pi}\mu$, where $\mu =\sqrt{\tfrac{n}{2\pi}}(\operatorname{e}^{2\cal R}-1)^{-\frac{1}{2}}$.
\end{lemma}
\begin{proof}
	Let us first define  $g(x)= \sqrt{n}\tfrac{C(\rho)-\cal R}{\sqrt{V(\rho)}}$, then we resort to a linearization of the $\operatorname{Q}$-function~\cite{makki2015finite},~\cite{makki2014finite}
	\begin{align}\label{eq:W(t)}
	K(t) \approx \operatorname{Q}(g(x))\!=\!
	\left\{
	\begin{array}{@{}ll@{}}
	\ 1 & t\leqslant\varrho\\
	\ \dfrac{1}{2}-\dfrac{\mu}{\sqrt{2\pi}}(x-\theta) & \varrho<t< \vartheta\\
	\ 0 & t\geq \vartheta
	\end{array}\right.
	\end{align}
	where $\theta\!=\!2^{\cal R}-1$, $\vartheta\!=\!\theta\!+\!\sqrt{\tfrac{\pi}{2}\mu^{-2}}$, $\varrho=\theta\!-\!\sqrt{\tfrac{\pi}{2}\mu^{-2}}$.
	Then, we calculate $\operatorname{E_{X}}[\operatorname{Q}\left(g(x)\right)]\!=\! \int_{0}^{\infty}K(t)f_{X}(x)dx\nonumber$, where $f_{X}(x)$ is the probability density function (PDF) of the SNR of the link $X$, and the solution is given in (\ref{outage_fading_linear}). 
\end{proof}
\vspace{-4mm}
\begin{center}
	\begin{figure}[!t] 
		\centering
		\includegraphics[width=0.8\columnwidth,height=2.2in]{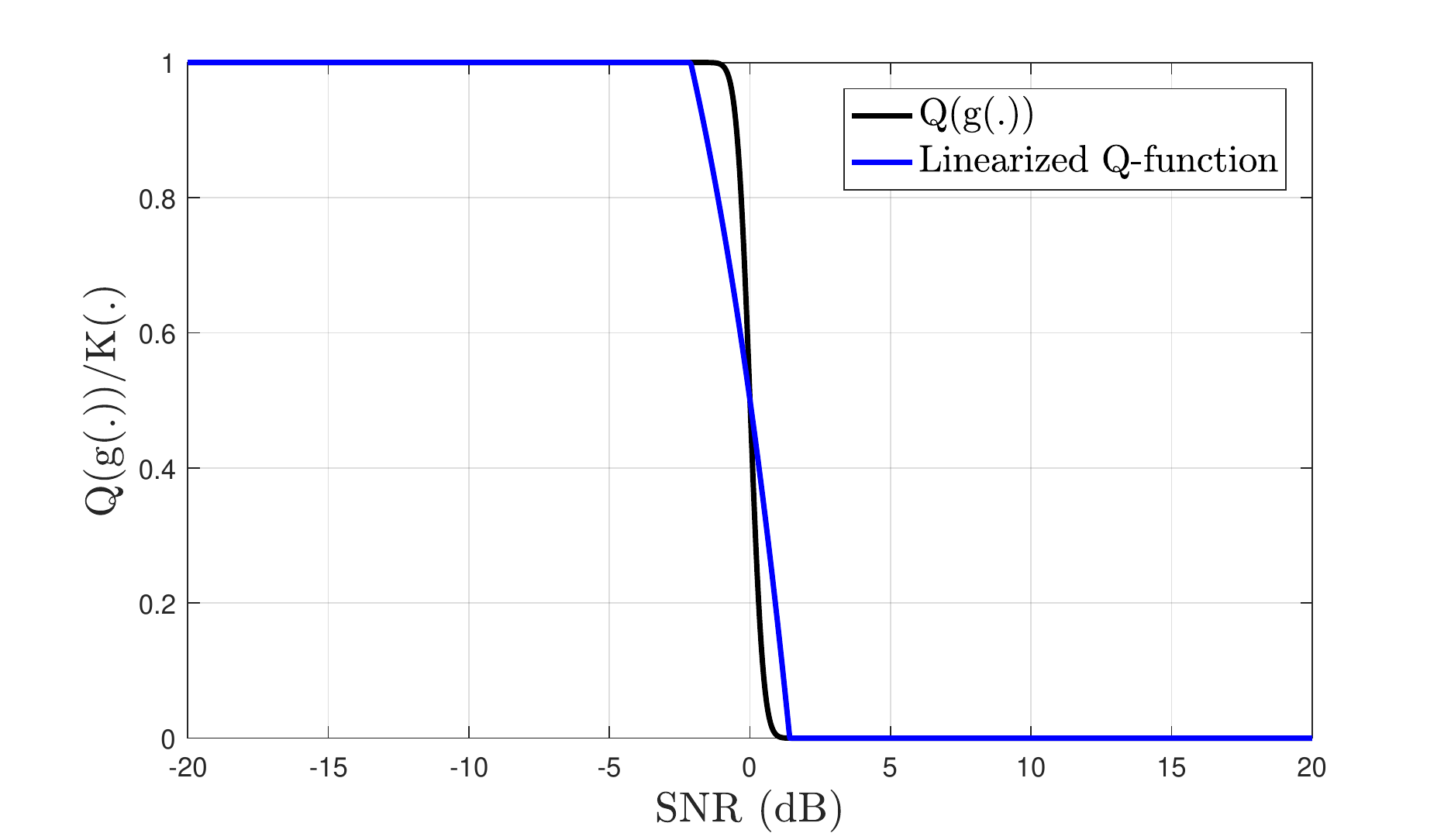}
		\vspace{-0mm}
		\caption{Accuracy of linearized Q-function compared to the original Q-function with $\cal R$$=1$ (bpcu). This figure illustrates the accuracy of Q-function linearization applied in closed form expression of the outage probability.}
		\label{fig:lemma1}
		\vspace{-0mm}
	\end{figure}
\end{center}
\begin{remark}
	 Moreover, we compare the accuracy of linearized Q-function in (\ref{eq:W(t)}) to original Q-function in (\ref{eq:outage probability}) as indicated in Fig.$2$. The difference between these two plots does not have a noticeable impact on the outage probability since we find the approximated outage probability in (\ref{outage_fading_linear}) via integrating over the SNR range and due to the symmetric property of the function as evinced by Fig.$2$, regions that show the difference between the original and linearized Q-function, cancel each other and so, this difference becomes negligible as illustrated in Fig.$3$. Thus, we can notice that error defined by $\text{error}=|\frac{\epsilon-\epsilon_{\text{app}}}{\epsilon}|$ is approximately equal to zero which shows the accuracy of the linearized Q-function applied in the closed form expression of the outage probability. Fore example, the maximum error over the entire SNR range is about $0.03 \%$. Similar conclusion holds for other
	values of $\cal R$.
\end{remark}

\begin{center}
	\begin{figure}[!htp] 
		\centering
		\includegraphics[width=0.8\columnwidth,height=2.2in]{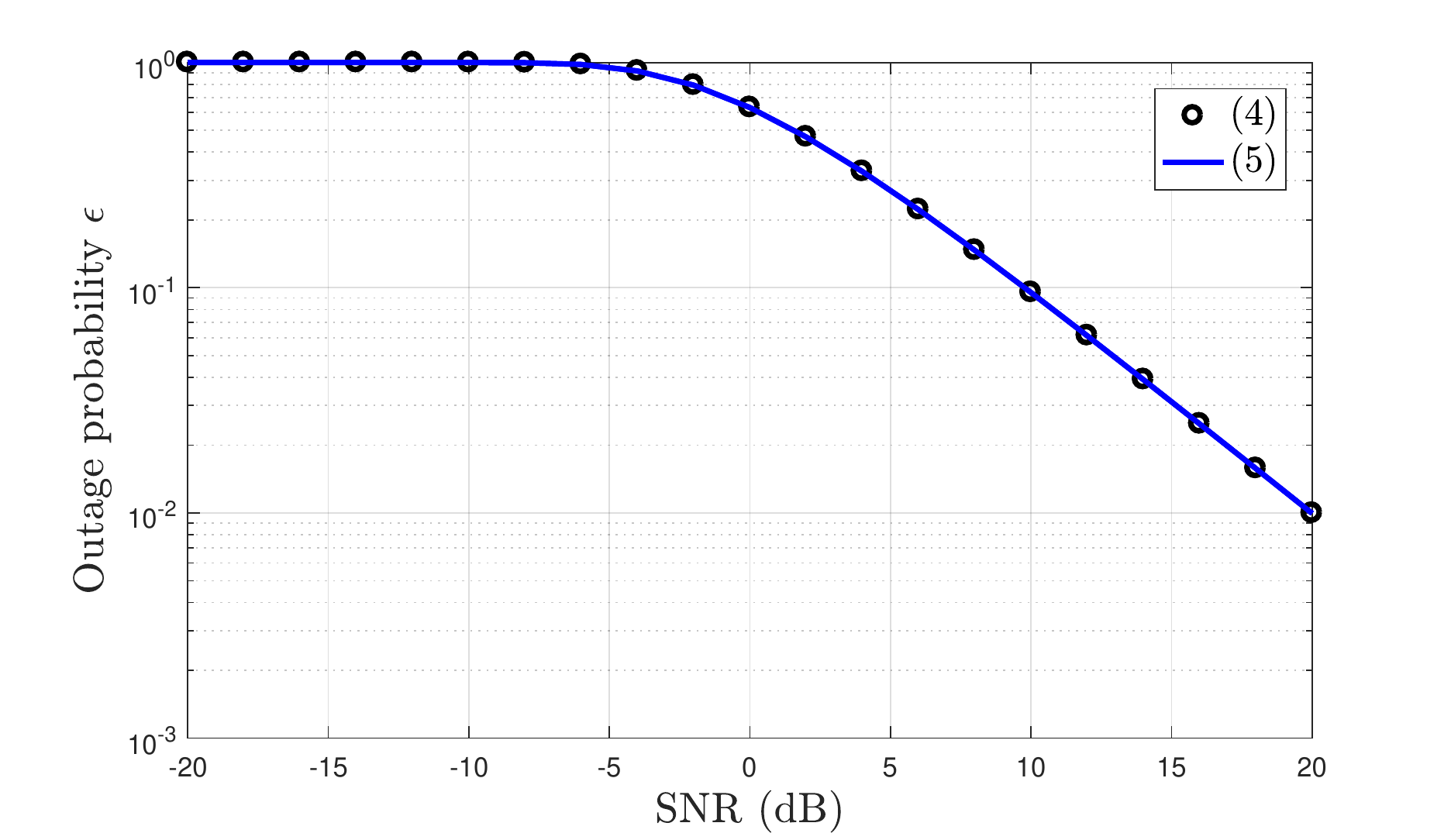}
		\vspace{-0mm}
		\caption{Comparing the accuracy of approximated outage probability in (4) and (5) with $\cal R$$=0.2$ (bpcu). This figure compares the accuracy of outage probability with the numerical integral of Q-function in (4) to the approximated outage probability in (5).}
		\label{fig:lemma11}
		\vspace{-0mm}
	\end{figure}
\end{center}

\section{The  Proposed  Method} \label{sec:Relaying assumptions}
In this section, we investigate the outage probability of cooperative DF, SC and MRC protocols under the FB regime. The direct transmission model is used here as the basis of the comparison. 
\subsection{Direct Transmission}
The source sends the message directly to the destination, where $\Omega_{Z'}\!=\!Z/\eta=P|h_{SD}|^2/N_0$, with average SNR $\gamma_{Z'}\!=\!P/N_0$, where the outage probability is calculated as in (\ref{outage_fading_linear}) but $\zeta$ with $P_{S}\!=\!P$ and $\mu$ with $n\!=\!n_{S}$. 

\subsection{Dual Hop Decode-and-Forward (DF)}
In this scheme, since the $S$-$D$ distance is too large, it assumes that the direct link is in the outage; thus, $R$ always collaborates with the source. Hence, $S$ sends the message to both $R$ and $D$ in the broadcasting phase. Then, $R$ transfers the message to $D$~\cite{laneman2004cooperative}. The overall outage probability is given by
\begin{equation} \label{DF outage}
\epsilon_{DF}\!=\!\epsilon_{SR}+(1-\epsilon_{SR})\epsilon_{RD},	
\end{equation}
where $\epsilon_{SR}$ and $\epsilon_{RD}$ are calculated according to (\ref{outage_fading_linear}). Notice that we update $\zeta$ with $P_{S}\!=\!\eta P$, $P_{R}= (1\!-\!\eta) P$ and $\mu$ with $n\!=\!n_{S}$, $n\!=\!n_{R}$, respectively. This scenario can be analyzed as selection combining (SC) or maximum ratio combining (MRC) depending on how the destination combines the original transmitted signal and the retransmitted signal.

%
\subsection{Selection Combining (SC)}
 In this protocol, $R$ starts to collaborate with $S$ if the destination confirms that the source transmission was unsuccessful and so, the destination requests for retransmission from $R$ to receive the frame correctly. Cooperation occurs if $R$ decodes the received message from $S$ correctly and so, transfers the message to $D$. Thereafter, if $D$ confirms that the transmission from $R$ is also failed, $D$ requests for the next subsequent message from $S$. Thus, the outage probability happens only if both $S$-$D$ and $R$-$D$ links are in outage~\cite{alves2012throughput},~\cite{alves2011performance}. The overall outage is given by
\begin{equation}\label{eq:SDF_outage}
\epsilon_{SC}\!=\!\epsilon_{SD}\epsilon_{SR}+(1-\epsilon_{SR})\epsilon_{SD}\epsilon_{RD},
\end{equation} 
where $\epsilon_{SD}$ is equal to (\ref{outage_fading_linear}) where $\zeta$ is updated with $P_{S}\!=\!\eta P$ and $\mu$ with $n\!=\!n_{S}$.

%
\subsection{Maximum Ratio Combining (MRC)}\label{sec:MRC}
In this scenario, relay always collaborates with the source and so, the channel gains of $S$-$D$ and $R$-$D$ links are combined at the receiver. Thus, the aggregated SNR is bigger than the primary attempted transmission rate as the $S$-$D$ transmission failed. In addition, the outage probability occurs if $S$-$D$ or $R$-$D$ transmission fails. Hence, the instantaneous SNR is $\Omega_{W}\!=\! \Omega_{Z}\!+\!\Omega_{Y}$~\cite{alves2012throughput},~\cite{alves2011performance}.
The  outage probability is~\cite{alves2012throughput}
\begin{equation}\label{eq:MRC_outage}
\epsilon_{MRC} \!=\!\epsilon_{SD}\left(\epsilon_{SR}\!+\!\left(1\!-\!\epsilon_{SR}\right)\frac{\epsilon_{SRD}}{\epsilon_{SD}}\right),
\end{equation}
where $\epsilon_{SRD}$ is the outage probability of the source-to-relay-to destination link, notice that the term $\tfrac{\epsilon_{SRD}}{\epsilon_{DF}}$ refers to the probability that D was not able to decode S message alone.
In order to calculate the (\ref{eq:MRC_outage}), first we need to attain the PDF of $W$, and then we calculate the outage probability as proposed in proposition~\ref{propose1}. To do so, let $W$ denote the sum of two independently distributed exponential random variables (RV), $Z$ and $Y$. Then, $f_W (w)$ is~\cite{alves2012throughput} 	
	\begin{equation}\label{eq:MRC_PDF}
	f_W{(w)} \!=\!
	\left\{
	\begin{array}{@{}ll@{}}
	\ \dfrac{w}{\Omega_{Z}^2} \operatorname{exp}\left(-\dfrac{w}{\Omega_{Z}}\right)&\Omega_{Z}\!=\!\Omega_{Y}\\
	\ \dfrac{\operatorname{exp}\left( -\dfrac{w}{\Omega_{Z}}\right)\!-\!\operatorname{exp}\left(- \dfrac{w}{\Omega_{Y}}\right) }{\Omega_{Z}\!-\!\Omega_{Y}} &\Omega_{Z}\!\neq\! \Omega_{Y}\\
	\end{array}\right.
	\end{equation}
		Since the RVs are independent, the proof is straightforward solution of 
		$f_W(w)\!=\! \int_0^\infty  f_Z(w-y) f_y(y) \mathsf{d} y$~\cite{athanasios2017probability}. 
	\begin{proposition}\label{propose1}
		The outage probability of the MRC of the $S$-$D$ and $R$-$D$ links $\epsilon_{SRD}$, is equal to 
		%
		\begin{equation}\label{eq:Outage_MRC_overall}
		\epsilon_{SRD}\!=\! 
		\left\{
		\begin{array}{@{}ll@{}}
		\dfrac{2-\operatorname{e}^{-\varphi}\!+\! \operatorname{e}^{-\alpha}}{2}\!-\!\dfrac{\varrho \operatorname{e}^{-\varphi}}{\Omega_{Z}} \!+\!\dfrac{\varrho \operatorname{e}^{-\varphi}-\vartheta \operatorname{e}^{-\alpha}}{2\Omega_{Z}}\!+\! \dfrac{\mu\theta \varDelta +2\mu\xi}{\sqrt{2\pi}}\!+\!\dfrac{(\mu\tau)/\Omega_{Z}}{\sqrt{2\pi}}&\Omega_{Z}\!=\!\Omega_{Y}\\
		\\
		\
		\dfrac{\Omega_{Z}\!-\!\Omega_{Y}\!+\!\Omega_{Z}\operatorname{e}^{-\alpha
			}\lambda_{1}\!+\!\Omega_{Z}\operatorname{e}^{-\varphi}\lambda_{2}\!+\!\Omega_{Y}\operatorname{e}^{-\tfrac{\vartheta}{\Omega_{Y}}}\lambda_{3}\!+\!\Omega_{Y}\operatorname{e}^{-\tfrac{\varrho}{\Omega_{Y}}}\lambda_{4}}{\Omega_{Z}\!-\!\Omega_{Y}}&\Omega_{Z}\!\neq\!\Omega_{Y}\	
		\end{array}\right.
		\end{equation}
		%
		\begin{equation}
		\tau\!=\!\vartheta^{2}\operatorname{e}^{-\tfrac{\vartheta}{\Omega_{Z}}}-\varrho^{2}\operatorname{e}^{-\tfrac{\varrho}{\Omega_{Z}}}-\theta\vartheta \operatorname{e}^{-\tfrac{\vartheta}{\Omega_{Z}}}+\theta\varrho \operatorname{e}^{-\tfrac{\varrho}{\Omega_{Z}}},
		\end{equation}
		\begin{equation}
		\xi\!=\!\vartheta \operatorname{e}^{-\tfrac{\vartheta}{\Omega_{Z}}}-\varrho \operatorname{e}^{-\tfrac{\varrho}{\Omega_{Z}}}+\Omega_{Z} \operatorname{e}^{-\tfrac{\vartheta}{\Omega_{Z}}}-\Omega_{Z} \operatorname{e}^{-\tfrac{\varrho}{\Omega_{Z}}}, 
		\end{equation}	
		where, $\vartheta$ and $\varrho$ are specified in (\ref{eq:W(t)}), and $\lambda_{1}\!=\!\tfrac{\mu\vartheta+\mu \Omega_{Z}-\mu\theta}{\sqrt{2\pi}}\!-\!\tfrac{1}{2}$, $\lambda_{2}\!=\!\tfrac{\mu\theta-\mu\varrho-\mu \Omega_{Z}}{\sqrt{2\pi}}\!-\!\tfrac{1}{2}$, $\lambda_{3}\!=\!\tfrac{1}{2}\!-\!\tfrac{\mu\vartheta+\mu \Omega_{Y}-\mu\theta}{\sqrt{2\pi}}$, $\lambda_{4}\!=\!\tfrac{1}{2}\!+\!\tfrac{\mu\varrho+\mu \Omega_{Y}-\mu\theta}{\sqrt{2\pi}}$, $\varDelta = \operatorname{e}^{-\varphi}\!-\!\operatorname{e}^{-\alpha}$, $\varphi = \tfrac{\varrho}{\Omega_{Z}}$ and $\alpha=\tfrac{\vartheta}{\Omega_{Z}}$.
	\end{proposition}
	\begin{proof}
		By plugging (\ref{eq:MRC_PDF}) into (\ref{outage_fading_linear}) and multiplying by the linearized Q-function $K(t)$, we attain
		\begin{equation}
		\epsilon\!=\!\int_{0}^{\varrho}f_X{(x)} dx+ \int_{\varrho}^{\vartheta} \left(\dfrac{1}{2}- \frac{\mu}{\sqrt{(2\pi)}}\left(x-\theta\right)\right)f_X{(x)}dx,
		\end{equation}
		which is solved with help of~\cite[Eq. 2.321]{gradshteyn2014table} 
		and after some algebraic manipulations we attain (\ref{eq:Outage_MRC_overall})~\cite{nouri2017performance}.	
	\end{proof}
\subsection{Asymptotic Analysis}
The outage probability in (\ref{outage_fading_linear}) can be defined as $\operatorname{P}[\Omega_i\!\leq\! \gamma_{\text{th}}]$ as the SNR goes to infinity, where $\gamma_{\text{th}}\!=\!2^{\cal R}\!-\!1$ . Thus, the approximated asymptotic outage probability per link in Rayleigh fading channels is $\epsilon_i\!=\!1 \!-\!\exp(-\gamma_{\text{th}}/\gamma_{i})$, where $\gamma_i$ is a function of $\beta$ and $P$~\cite[\S 10]{dosti2017ultra}. Thereafter, we resort to Taylor series as $\gamma_\text{th}/\gamma_i$ approaches zero as $\text{SNR}\rightarrow \infty$, and so, $\exp(x)\!\approx\! 1+x$ and attain an asymptotic expression as $\epsilon_i \approx (\gamma_{\text{th}}/\gamma_{i})$~\cite[\S 11]{dosti2017ultra}. The asymptotic expression of $\epsilon$ after maximum ratio combining of $S$ and $R$ transmissions $\epsilon_{\text{SRD}}$ is given in \cite[\S 7]{alves2012throughput}. Therefore, the outage probability is approximated as  $\!\epsilon_{\text{SRD}}^\infty\!\approx\!\frac{\gamma_{\text{th}}^2}{2 \gamma_{Z} \gamma_{Y}}$, resorting to series expansion as $\gamma_i \rightarrow \infty$. 

\section{Numerical Results and Discussion}
\subsection{URLLC via Cooperative Diversity}\label{URLLC}
In this section we show some numerical results of cooperative relaying transmission under FB regime. First, we show the impact of coding rate on the probability of successful transmission where MRC protocol outperforms DF, SC and DT in terms of reliability. We also indicate the minimum latency required to support URLLC. Thereafter, we show the optimal value of power allocation factor $\eta$ for each of studied protocols. In addition, we compare the performance of cooperative relaying to DT in terms of power consumption and blocklength to perform under the UR region (URR). We verify the accuracy of our analytical model through the Monte-Carlo simulations. Unless stated; otherwise, assume maximum transmit power per link as $20$ dB, $n =500$, $k=500$, and $R$ is in between $S$ and $D$, with $\beta = \frac{1}{2}$. The URR is shaded purple area in the following plots, and its most loose constraint is denoted with a red line where the outage probability is $10^{-3}$, thus $ 1- \epsilon =99.9 \%$ reliability is feasible.

\subsubsection{Reliability vs. Coding Rate}
Fig.$4$ compares the probability of successful transmission ($P_{succ} = 1-\epsilon$) as a function of coding rate in URR. We can clearly see that MRC supports URLLC with higher coding rates compared to DT, DF and SC schemes which is more evident with short packet lengths under the FB regime. Hence, MRC is less affected by the coding rate growth under the URR. For instance, with $n=400$ and $\cal R$$= 0.43$, MRC covers $99.999 \%$ reliability, while SC provides equal reliability as MRC but with lower coding rate as $k= 153$ and $n=400$. In addition, with $n =400$, reliability decreases to $99.9 \%$ and $99 \%$ with $k=133$ and $k=293$ for DF and DT schemes, respectively. Thus, URLLC is feasible via the cooperative schemes and we can apply each of these schemes based on our requirements such as reliability, packet length and number of transmitted information bits. 

\begin{figure*}[!t]
	\centering
	\includegraphics[width=\columnwidth,height=3.25in]{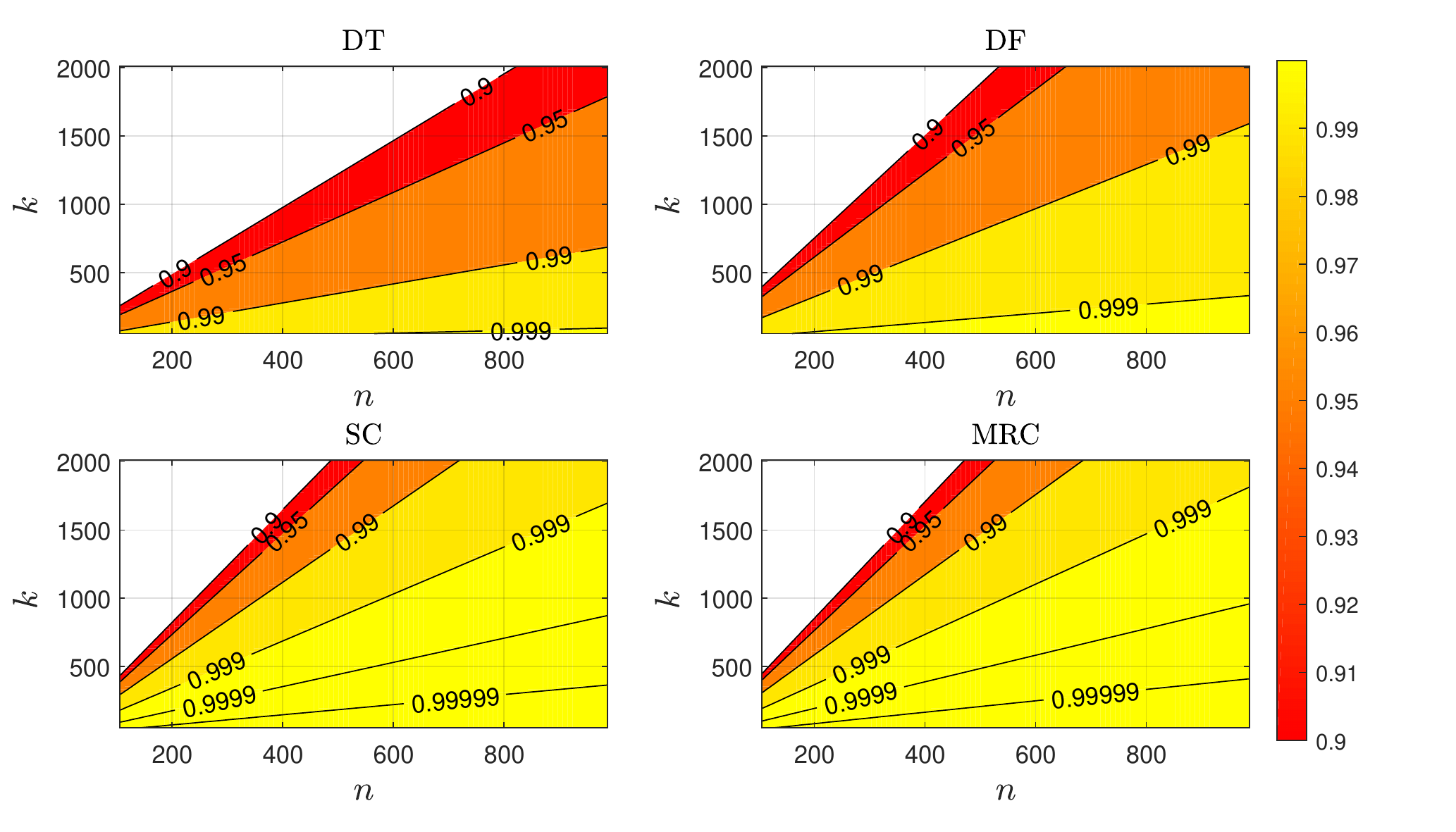}\label{fig:Psucc}
	\caption{Probability of successful transmission as a function information bits $k$ and blocklength $n$. This figure compares the impact of coding of on the performance of cooperative schemes, namely DF, SC and MRC to direct transmission. We show that MRC outperforms other studied cooperative scenarios in terms of reliability and coding rate.}
\end{figure*}
	
In Fig.$5$ we examine the impact of power allocation factor $\eta$ on the probability of successful transmission. As mentioned earlier in Section \ref{systemmodel}, in order to provide a fair comparison between DT and cooperative schemes, we allocate powers to $S$ and $R$ according to the power allocation factor. In DF, outage probability is minimized via equal power allocation strategy while in SC and MRC, we exploit additional diversity of the direct link; thus, less power should be allocated to $R$ as shown in Fig.$5$. We also illustrate that URLLC is feasible through the cooperative schemes, particularly with SC and MRC where the outage probability is minimized to $0.2 \%$ and $0.1 \%$, respectively. As we indicated in our previous work in~\cite{nouri2017performance}, these results holds for other values of SNR and coding rate.
\begin{figure}[!t]
	\centering
	\includegraphics[width=1\columnwidth,]{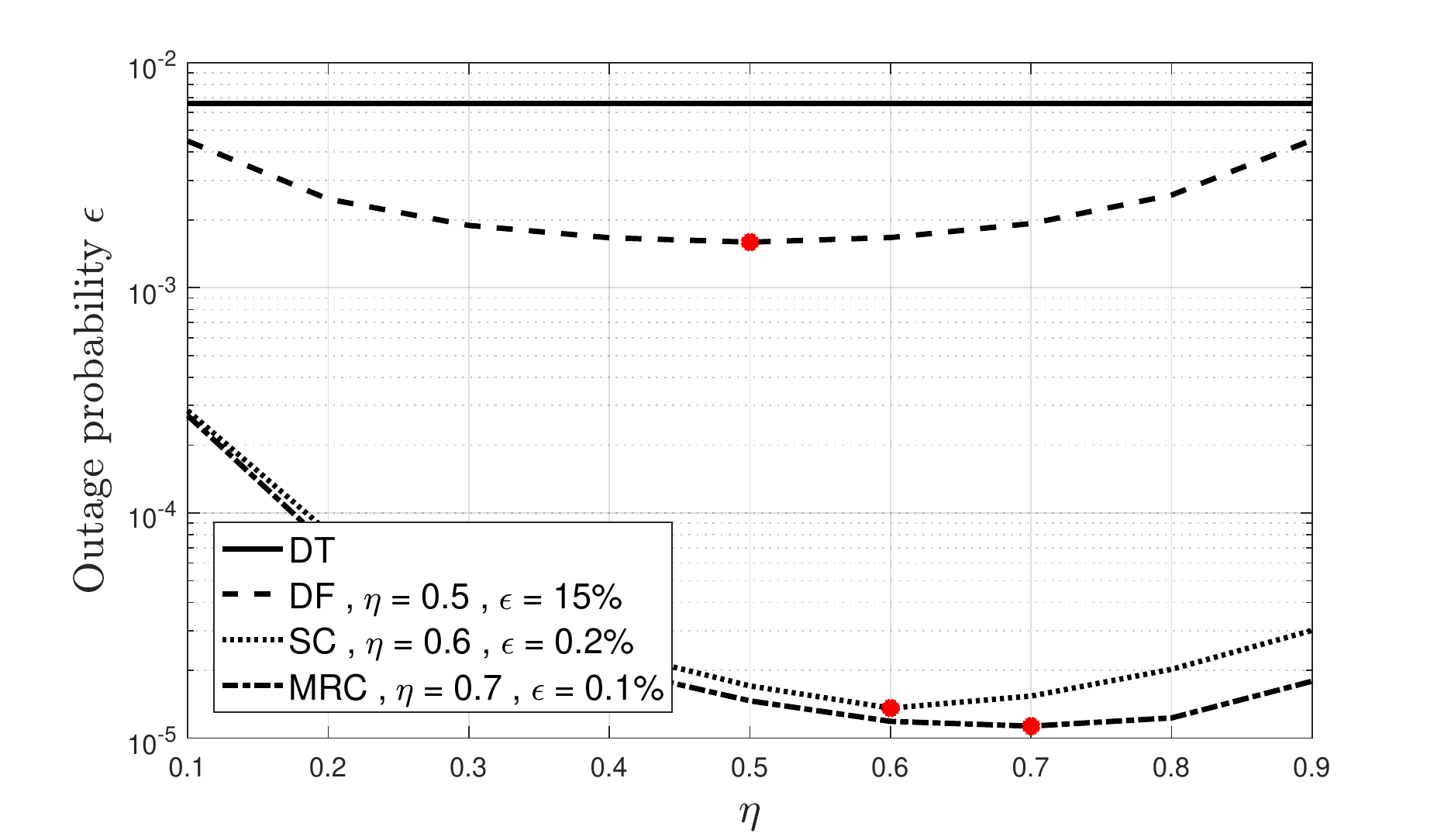}\label{fig:eta}
	\caption{Impact of power allocation on the outage probability with $k=250$, $n =500$ and maximum SNR=$20$ dB. This figure illustrates the optimal value of power allocation factor for each of studied relaying scenarios, where the outage probability is minimized. We show that we have to allocate more power to the source to work in ultra-reliable region, and MRC requires more power at the source compared to the other studied protocols. }
\end{figure}

\begin{figure}[!t]
	\centering
	\includegraphics[width=1\columnwidth,]{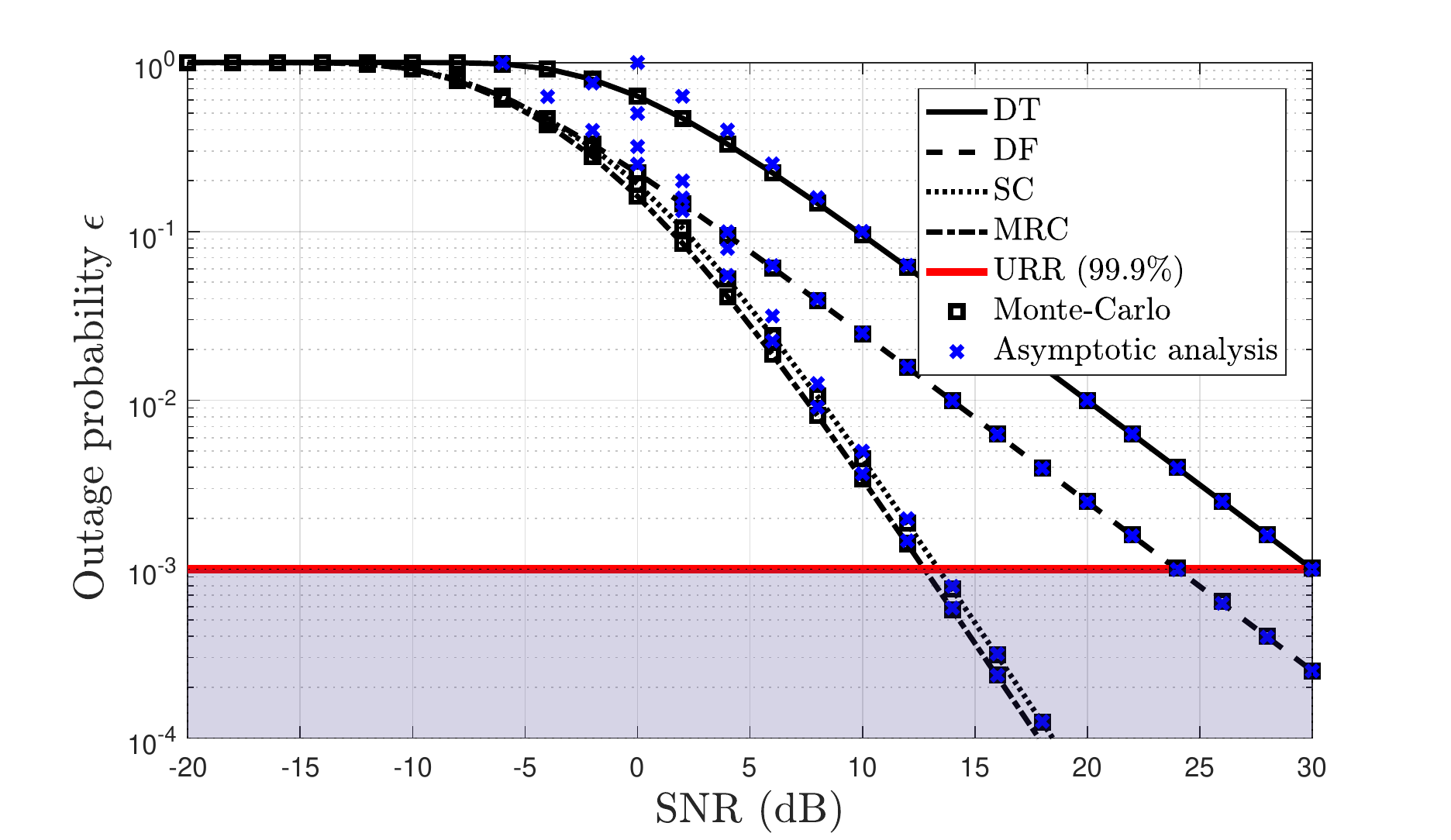}\label{fig:powerconsumption}
	\caption{Performance analysis of cooperative protocols with $k=500$ and $n=500$ under EPA. This figure compares the power consumption of studied cooperative schemes to direct transmission under the equal power allocation strategy. We show that under EPA, MRC requires less transmit power to work under ultra-reliable region compared to the other studied protocols. Monte-Carlo simulations confirm the accuracy and appropriateness of our analytical model. In addition, asymptotic expressions approach the analytical results as the SNR increases.}
\end{figure}

\begin{figure}[!t]
	\centering
	\includegraphics[width=1\columnwidth,]{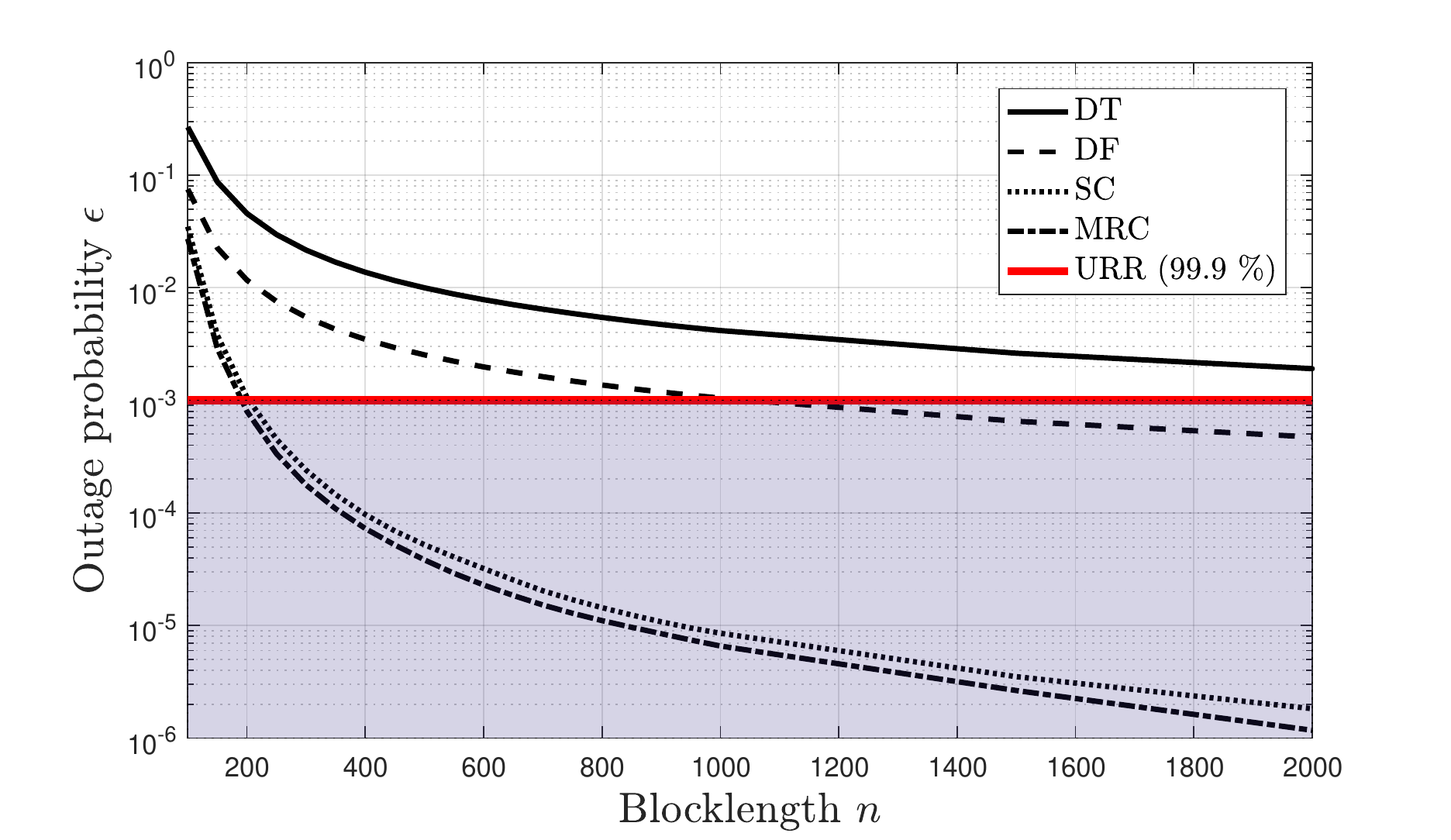}\label{fig:blocklength}
	\caption{Outage probability analysis for DT, DF, SC and MRC under EPA. This figure compares the performance of cooperative schemes to direct transmission in terms of outage probability under the finite blocklength regime. We indicate that URLLC is feasible via cooperative diversity technique. We show the superiority of MRC and SC over DF and DT schemes under the finite blocklength regime.}
\end{figure}

In Fig.$6$ we compare the ultra-reliable performance of cooperative schemes to DT in terms of transmit power under equal power allocation constraint. We can clearly see the power gain attained via the cooperative protocols at high SNR regime where there is huge performance gap between cooperative schemes and DT. In addition, we indicate that MRC and SC protocols perform closely in the entire SNR range and consumes less transmit power to communicate under the URR in comparison to DF and DT. 

In addition, we indicate the possibility of using asymptotic expressions in ultra-reliable region. In other words, at high SNR regime, the maximum achievable coding rate (\ref{eq:maximum rate}) converges the asymptotically long codewords as $\cal R$$_\text{asym}(n,\epsilon)= C_\epsilon$, where $C_\epsilon=\text{sup}\{{\cal R}:\text{Pr}[\log_2(1+\rho)< {\cal R}]<\epsilon\}$. In Fig.$6$ we show that the asymptotic expressions approach the analytical results as the transmit power increases.

Fig.$7$ indicates the performance advantage of cooperative schemes over DT. Cooperative schemes exploit diversity gain which decreases the outage probability remarkably. As we expected, the outage probability decreases in blocklength. In addition, SC and MRC protocols are able to support URLLC under FB regime with very short packet lengths. 

Fig.$8$ indicates the total minimum latency required for URLLC under the FB regime with two distinct power constraints as $i$) EPA: $ P =\eta P+ (1-\eta)P$, where $\eta = 0.5$, and $ii$) OPA: $P_S +P_R \leq P$. The choices of the minimum latency $\delta$ and optimal powers are in such a way that minimizes the outage probability constraint to a specific interval of interest and holds the power constraints which gives the optimal values of $n$ and $P$, and is a nonlinear optimization problem as follows \footnote{We solve the optimization problem numerically with the Matlab function $fmincon$. Interior point algorithm is used to solve the nonlinear optimization problem~\cite{waltz2006interior}.}.
\begin{equation*}
\begin{aligned}
\centering
& \underset{}{\text{minimize}}
& &  \epsilon (n,P) \\
& \text{subject to}
& & P_S + P_R \leq P,\\
& & & 100\leq n\leq10000.
\end{aligned}
\end{equation*}
We set the minimum  blocklength to 100 since (\ref{outage_fading}) is accurate for $n > 100$, as proved for AWGN channels \cite[Figs. $12$ and $13$]{polyanskiy2010channel} as well as for fading channels as discussed in~\cite{yang2014quasi}, and to a maximum of $10^4$ so to reduce the search range, and to be within URLLC boundaries. 

It shows that DT is not able to cope with the stringent latency constraint and need a large tolerance of delay; thus, we resort to cooperative protocols in order to reduce the latency in URR. It can be clearly seen that SC works highly better than DF and performs closely to MRC in the entire range but with higher latency requirements when we allocate equal powers to the $S$ and $R$. We also indicate that under OPA constraint, SC outperforms MRC in terms of channel uses and is more energy efficient than MRC as we discuss about it in the following section, while with equal power allocation strategy, MRC requires less channel uses and consumes less transmit power as we can see in Fig.$6$. Here, with equal power allocation strategy, the total power of $S$ and $R$ may be less than the maximum total power ($20$ dB)
but in Fig.$6$ we force $P_S$ and $P_R$ to be equal with total power of $P_{max}$. Therefore, according to the simulations, when $\cal R$$=1$ bpcu and $P=20$ dB, higher reliability is feasible in Fig.$6$ compared to Fig.$8$.

\begin{figure}[!t]
	\centering
	\includegraphics[width=1\columnwidth,]{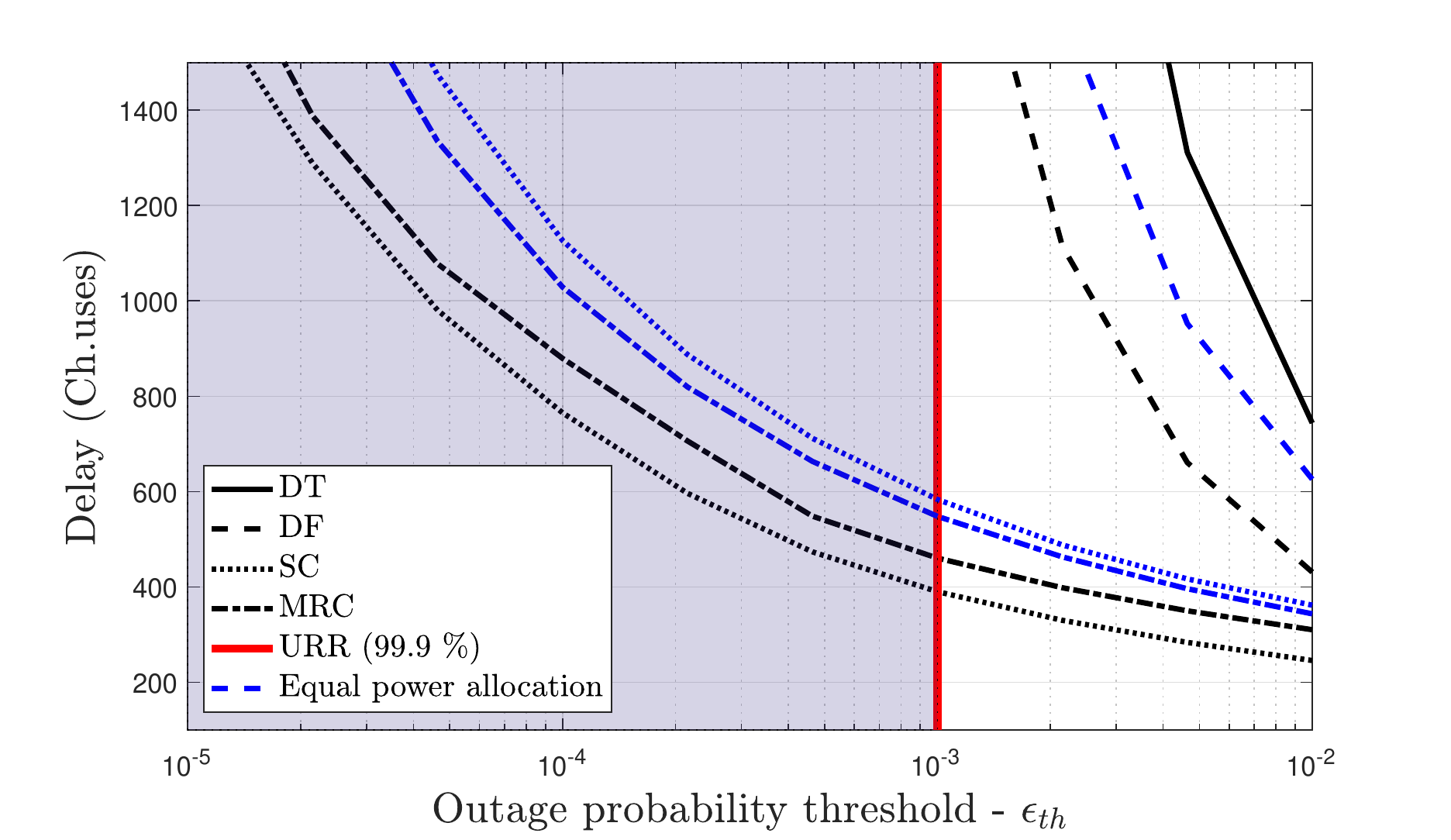}\label{fig:delay}
	\caption{Total latency in terms of channel uses with EPA and OPA constraints.
	This figure compares the total minimum latency in cooperative schemes under two distinct power constraints, so-called EPA and OPA. We show that under EPA, MRC has the lowest latency while with OPA, SC outperforms DF and MRC protocols in terms of latency requirement.}
\end{figure}

\subsection{Energy Efficiency Analysis}\label{EE}
Energy efficiency (EE) determines the trade-off between throughput gains and total energy consumed. Let us first define the total energy consumption per bit of each scenario. The total power consumed includes power of transmission with no dependency on the distance of relay nodes, consumed power in radio frequency(RF) circuitry and also coding rate~\cite{de2011energy},~\cite{alves2014outage}. Here, we ignore the baseband processing consumption since its value is negligible in comparison to the energy consumption of RF circuitry~\cite{cui2005energy}. 

Then the total energy consumption per bit of a single-hop transmission is
\begin{equation}
\operatorname{ E_{\operatorname{SH}}} = \dfrac{P_{PA}+P_{TX}+P_{RX}}{\cal R},
\end{equation}
where $P_{PA} = P/ \phi $ is the power amplifier consumption for a single-hop transmission and $\phi$ is the drain efficiency of the amplifier, $P_{TX}$ and $P_{RX}$ are the power consumed for transmitting and receiving in the internal circuitry, respectively.
In a similar way, we can also find the total power consumption of multi-hop schemes by determining the outage probability on $S$-$R$ link in each cooperative schemes.
\subsubsection{Cooperative Transmissions}
The total power consumption for DF protocol depends on the outage probability of $S$-$R$ link as follow
\begin{equation} \label{EE_DF}
\operatorname{ E_{\operatorname{DF}}} = \epsilon_{SR_{DF}} \times \dfrac{P_{PA}+P_{TX}+P_{RX}}{\cal R}
+(1-\epsilon_{SR_{DF}})\times \dfrac{2P_{PA}+2P_{TX}+2P_{RX}}{\cal R},
\end{equation} 
where the first term indicates that the consumed energy on the $S$-$R$ link, while the second term shows that $R$ could decode the message correctly and send the packet to $D$. 

In the case of SC and MRC, the total power consumption is formulated as follow
\begin{equation} \label{EE_SC}
\operatorname{ E_{\operatorname{j}}} = \epsilon_{SR_{j}} \times \dfrac{P_{PA}+P_{TX}+2P_{RX}}{\cal R}
+(1-\epsilon_{SR_{j}})\times \dfrac{2P_{PA}+2P_{TX}+3P_{RX}}{\cal R},
\end{equation}
where  $j\in \{SC, MRC\}$ and $\epsilon_{SR}$ is calculated by (\ref{outage_fading_linear}) accordingly to each method. The additional $P_{RX}$ in each term of (\ref{EE_SC}) compared to the (\ref{EE_DF}), corresponds to the transmission of $S$, which is heard by both $R$ and $D$ and destination decodes $S$-$D$ and $R$-$D$ transmissions, simultaneously.

Hence, the EE for each protocol is formulated as
\begin{equation}
EE = \dfrac{\textrm{Throughput}}{\textrm{Total energy consumed}}= \dfrac{{\cal R} (1-\epsilon)}{\operatorname{E}} , \;\;\;\;\; \big(\textrm{bits/Ch.uses/W}\big)
\end{equation}

	\begin{figure}[!t]
	\centering
	\includegraphics[width=1\columnwidth]{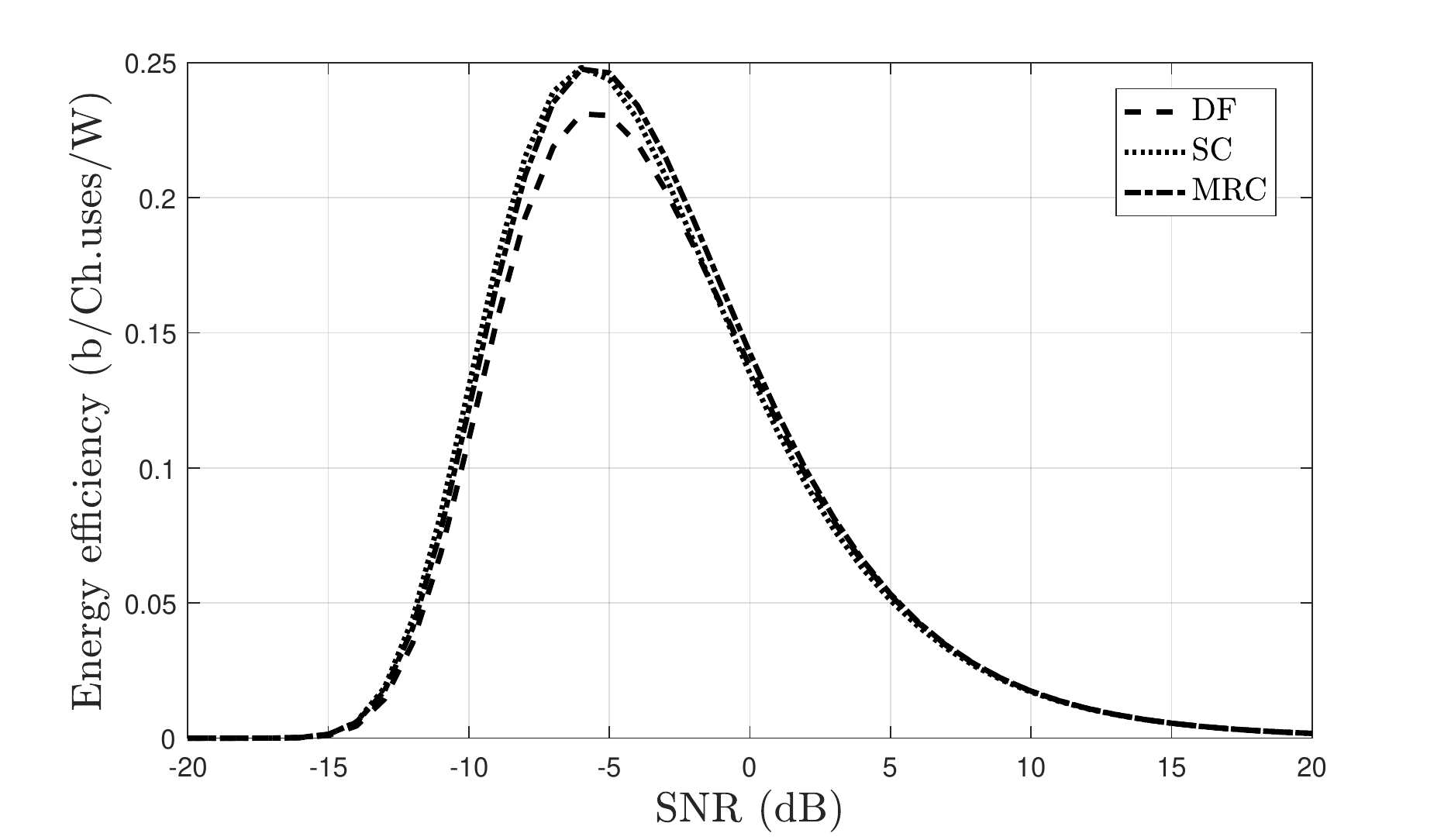}\label{F1_tput}
	\caption{Energy efficiency analysis of cooperative schemes under equal power allocation. This figure compares the energy efficiency in DF, SC and MRC schemes under EPA as a function maximum transmit power.}
\end{figure}

Furthermore, as observed from Fig.$9$, $EE(\epsilon,\operatorname{E})$ is non-convex in the SNR, while the outage probability is monotonically decreasing in the SNR and energy consumption is monotonically increasing, which is observed in Figs. $6$ and $11$, respectively. We maximize the energy efficiency as follow
	\begin{equation*}
	\begin{aligned}
	\centering
	& \underset{}{\text{Maximize}}
	& &  EE (\epsilon,E) \\
	& \text{subject to}
	& & P_S + P_R \leq P,\\
	& & & \epsilon(n,P) \leq \epsilon_{\text{th}},\\
	& & & 100\leq n\leq10000.
	\end{aligned}
	\end{equation*} 
	
	 This problem is equivalent to minimize the outage probability with respect to $P^*_S$, $P^*_R$ and blocklength $n^*$. Since we aim to compare the performance of cooperative schemes, we do not focus on the proposal of a particular solution, but we resort to numerically efficient algorithm. Therefore, we resort to $f_{mincon}$ implemented in Matlab and use interior point algorithm to solve the nonlinear optimization problem as detailed in~\cite{waltz2006interior}. We consider outage probability threshold in an interval of interest as $10^{-5}<\epsilon_{\text{th}}<10^{-2}$. At each outage probability value, we numerically determine $P^*_S$, $P^*_R$ and blocklength $n^*$ that maximize the energy efficiency. We apply the numerical optimization due to the nonlinear constraint on the outage probability $\epsilon(n,P) \leq \epsilon_{\text{th}}$.

\begin{figure}[!t]
	\centering
	\includegraphics[width=1\columnwidth,]{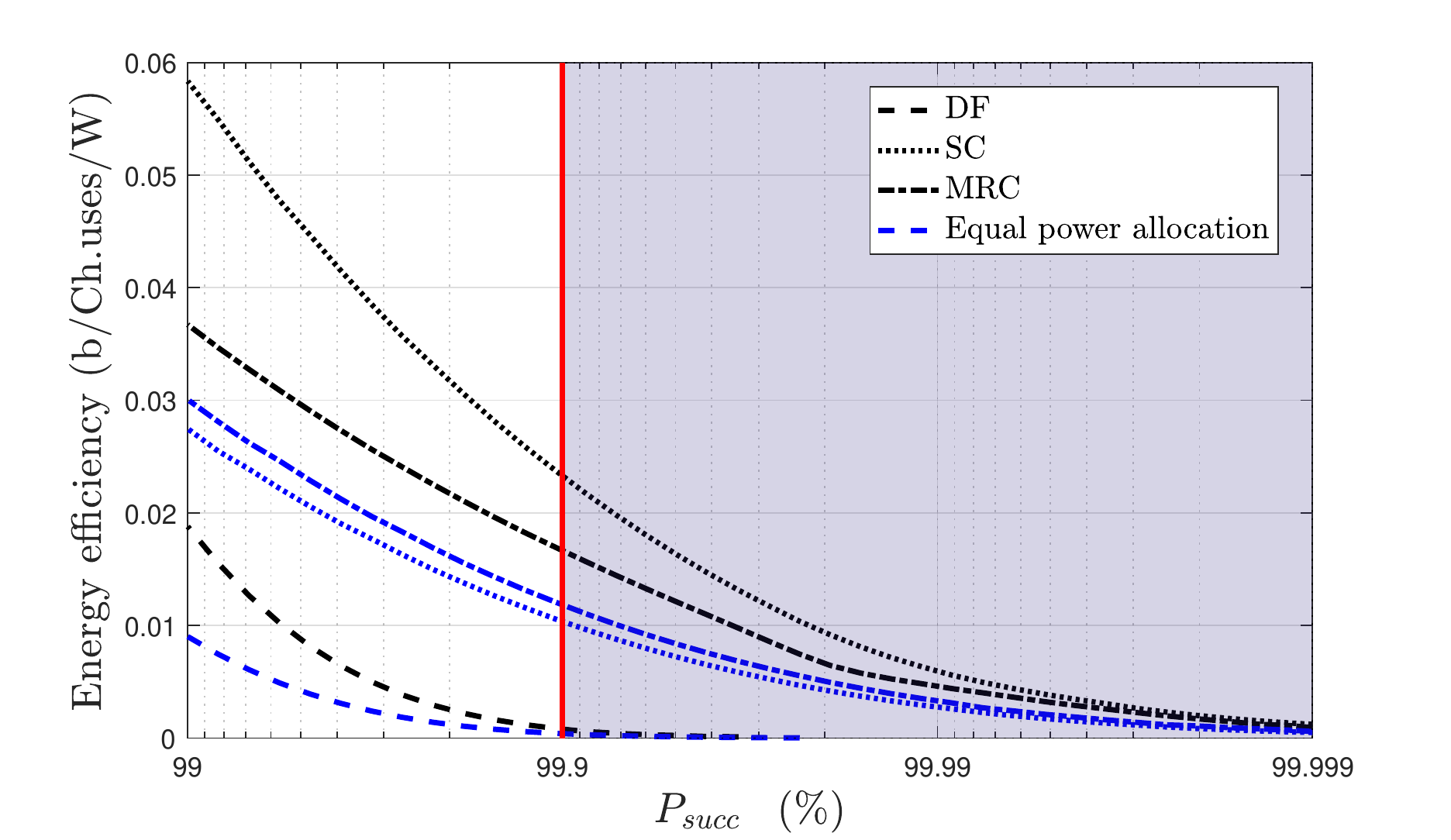}\label{fig:EE}
\caption{Energy efficiency of DF, SC and MRC scenarios under EPA and OPA strategies.
This figure compares the energy efficiency in DF, SC and MRC cooperative schemes under EPA and OPA strategies as a function reliability. We show that under EPA, MRC is the most energy-efficient strategy while with OPA strategy, SC becomes the most energy-efficient protocol.}
\end{figure}

\begin{figure}[!t]
	\centering
	\includegraphics[width=1\columnwidth,]{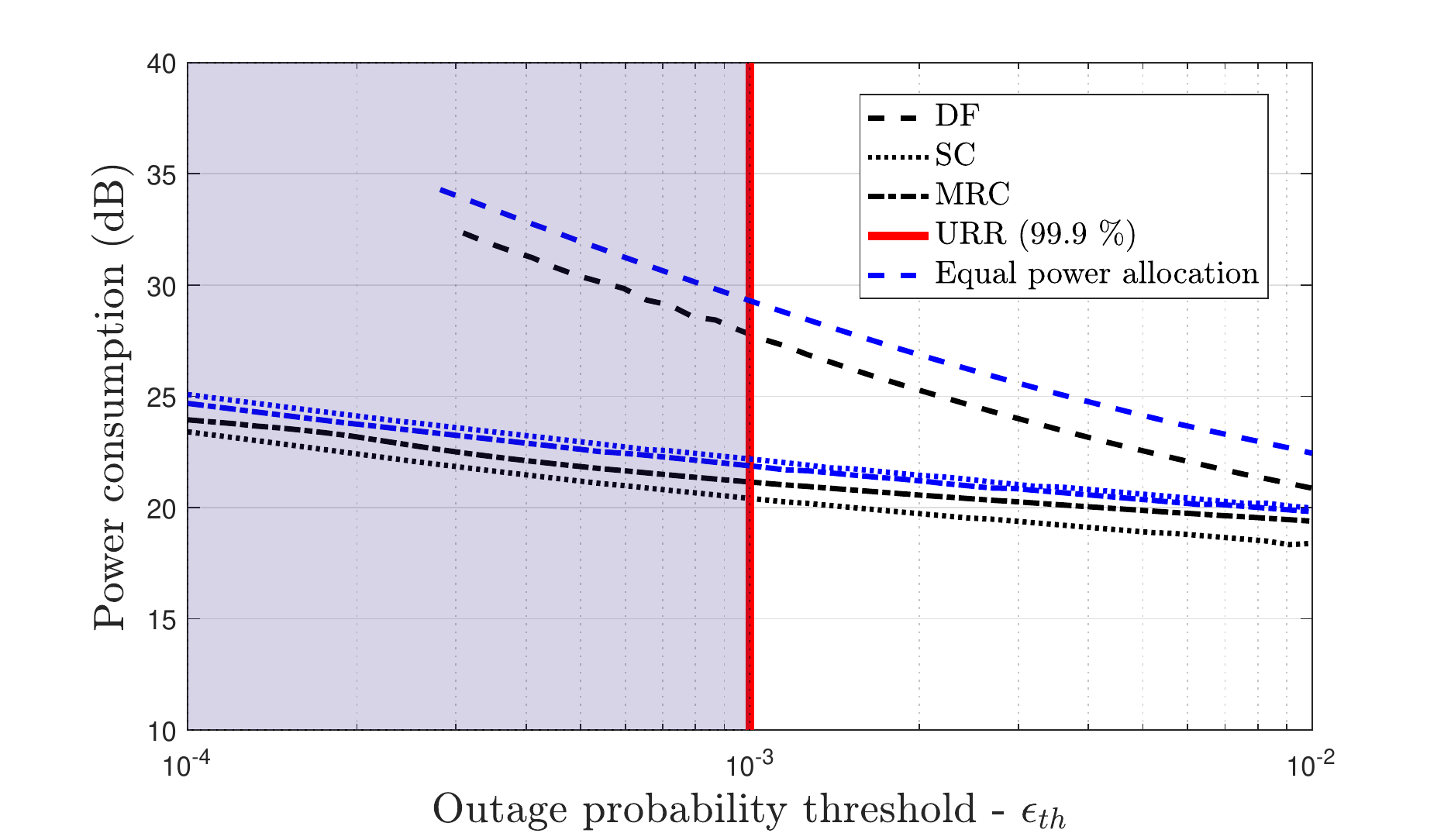}\label{fig:totalpower}
	\caption{Total power consumption of DF, SC and MRC under EPA and OPA constraints.
	This figure indicates the total power consumed in each of studied cooperative schemes. We show that under EPA, MRC consumes less energy compared to DF and SC scenarios while under OPA strategy, SC becomes the most energy-efficient scenario.}
\end{figure}

Fig.$10$ compares the energy efficiency of cooperative schemes in terms of probability of successful transmission under two distinct power constraints. In this paper, we assume $P_{TX}=97.9$ mW, $P_{RX}=112.2$ mW and the drain efficiency $\phi = 0.35$ according to the power consumption values reported in~\cite{cui2005energy}. Under EPA strategy, MRC is the most energy efficient scenario among other cooperative scenarios as it consumes less transmit power shown in Fig.$6$, and has lower latency in URR while under OPA, SC becomes the most energy efficient protocol as we show in Fig.$8$ it reduces the latency and the total power consumption is less than that of MRC. Since Fig.$5$ indicates that in order to perform in URR, we should allocate more power to the source where $\eta$ is equal to $0.6$ and $0.7$ for SC and MRC, respectively. Hence, more power is allocated to the source of MRC than that of SC; thus, MRC becomes less energy efficient compared to SC under OPA strategy.

Fig.$11$ compares the total consumed energy in each of studied cooperative scenarios under EPA and OPA strategies. Under EPA, as we expected MRC is superior and consumes less transmit power compared to DF and SC protocols, while with OPA, SC outperforms MRC and becomes most energy efficient protocol. In addition, with the maximum transmit power of $20$ dB, and no feasible solutions are found for DF protocol under stringent reliability requirements, which evinces the need for more sophisticated cooperative protocols. Feasible solutions are found if the transmit power increases, but it would be spectrally and energy insufficient. 

\section{Conclusions}\label{con}
In this paper, we assess the relay communication under the finite blocklength regime under Rayleigh fading. Performance of three relaying scenario, namely DF, SC and MRC is compared to direct transmission under two distinct power constraints so-called EPA and OPA. Based on the outage probability analysis of each transmission protocol, we show that relaying improves the probability of a successful transmission and guarantees ultra-high reliability with FB codes. MRC protocol is less affected by the coding and provide higher reliability compared to DT and two other relaying scenarios. In addition, we numerically show the optimal power allocation for the relaying protocols under study so to operate in URR. Our results shows that operation at URR is feasible by allocating more power to the source; however, relay node is considered to provide additional diversity gain compared to the DT which is more evident at high SNR regime. We compare the studied cooperative schemes in terms of latency and energy efficiency under the two distinct power constraints. According to the results, with equal power allocation at source and relay, MRC is the most energy efficient protocol with lower latency and power consumption compared to the other scenarios while SC has the highest energy efficiency and lowest latency under optimal power allocation strategy. Moreover, we provide the outage probability in closed form and prove the accuracy and appropriateness of our analytical model through numerical results. Finally, in our future work, we will focus on the impact of imperfect channel state information on URLLC.

\setcounter{secnumdepth}{0}
\section{Declarations}

\begin{backmatter}
\section*{Availability of data and material}
The  manuscript  is  self-contained.  Simulations  description  and  parameters  are  provided  in  details  in  Section  4.

\section*{Competing interests}
  The authors declare that they have no competing interests.
\section*{Funding}
This work has been partially supported by Finnish Funding Agency for Technology and Innovation (Tekes), Huawei Technologies, Nokia and Anite Telecoms, and Academy of Finland (under Grant no. 307492)

\section*{Author's contributions}
All  authors  have  contributed  to  this  manuscript  and  approved  the  submitted  manuscript.

\section*{Acknowledgments}
No applicable.
\section*{Author details}
Centre  for  Wireless  Communications  (CWC),  University  of  Oulu,  Finland
%

\begin{thebibliography}{10}
\bibitem{popovski2017ultra}
P.~Popovski, J.~J. Nielsen, C.~Stefanovic, E.~de~Carvalho, E.~Str{\"o}m, K.~F.
Trillingsgaard, A.-S. Bana, D.~M. Kim, R.~Kotaba, J.~Park \emph{et~al.},
``Ultra-reliable low-latency communication (urllc): Principles and building
blocks,'' \emph{arXiv preprint arXiv:1708.07862}, 2017.

\bibitem{shariatmadari2015machine}
H.~Shariatmadari, R.~Ratasuk, S.~Iraji, A.~Laya, T.~Taleb, R.~J{\"a}ntti, and
A.~Ghosh, ``Machine-type communications: current status and future
perspectives toward 5{G} systems,'' \emph{IEEE Communications Magazine},
vol.~53, no.~9, pp. 10--17, 2015.

\bibitem{mehmood2015mobile}
Y.~Mehmood, C.~G{\"o}rg, M.~Muehleisen, and A.~Timm-Giel, ``Mobile m2m
communication architectures, upcoming challenges, applications, and future
directions,'' \emph{EURASIP Journal on Wireless Communications and
	Networking}, vol. 2015, no.~1, p. 250, 2015.

\bibitem{kalor2017network}
A.~E. Kal{\o}r, R.~Guillaume, J.~J. Nielsen, A.~Mueller, and P.~Popovski,
``Network slicing for ultra-reliable low latency communication in industry
4.0 scenarios,'' \emph{arXiv preprint arXiv:1708.09132}, 2017.

\bibitem{andrews2014will}
J.~G. Andrews, S.~Buzzi, W.~Choi, S.~V. Hanly, A.~Lozano, A.~C. Soong, and
J.~C. Zhang, ``What will 5{G} be?'' \emph{IEEE Journal on selected areas in
	communications}, vol.~32, no.~6, pp. 1065--1082, 2014.

\bibitem{dawy2017toward}
Z.~Dawy, W.~Saad, A.~Ghosh, J.~G. Andrews, and E.~Yaacoub, ``Toward massive
machine type cellular communications,'' \emph{IEEE Wireless Communications},
vol.~24, no.~1, pp. 120--128, 2017.

\bibitem{lee2016packet}
B.~Lee, S.~Park, D.~J. Love, H.~Ji, and B.~Shim, ``Packet structure and
receiver design for low-latency communications with ultra-small packets,'' in
\emph{Global Communications Conference (GLOBECOM), 2016 IEEE}.\hskip 1em plus
0.5em minus 0.4em\relax IEEE, 2016, pp. 1--6.

\bibitem{taleb2012machine}
T.~Taleb and A.~Kunz, ``Machine type communications in {3GPP} networks:
potential, challenges, and solutions,'' \emph{IEEE Communications Magazine},
vol.~50, no.~3, 2012.

\bibitem{condoluci2015toward}
M.~Condoluci, M.~Dohler, G.~Araniti, A.~Molinaro, and K.~Zheng, ``Toward 5{G}
densenets: architectural advances for effective machine-type communications
over femtocells,'' \emph{IEEE Communications Magazine}, vol.~53, no.~1, pp.
134--141, 2015.

\bibitem{monserrat2015metis}
J.~F. Monserrat, G.~Mange, V.~Braun, H.~Tullberg, G.~Zimmermann, and
{\"O}.~Bulakci, ``Metis research advances towards the 5g mobile and wireless
system definition,'' \emph{EURASIP Journal on Wireless Communications and
	Networking}, vol. 2015, no.~1, p.~53, 2015.

\bibitem{bockelmann2016massive}
C.~Bockelmann, N.~Pratas, H.~Nikopour, K.~Au, T.~Svensson, C.~Stefanovic,
P.~Popovski, and A.~Dekorsy, ``Massive machine-type communications in 5{G}:
Physical and mac-layer solutions,'' \emph{IEEE Communications Magazine},
vol.~54, no.~9, pp. 59--65, 2016.

\bibitem{popovski2014ultra}
P.~Popovski, ``Ultra-reliable communication in {5G} wireless systems,'' in
\emph{1st International Conference
	on 5G for Ubiquitous Connectivity (5GU), 2014 }.\hskip 1em plus 0.5em minus 0.4em\relax IEEE, 2014, pp. 146--151.

\bibitem{singh2016selective}
B.~Singh, O.~Tirkkonen, Z.~Li, M.~A. Uusitalo, and R.~Wichman, ``Selective
multi-hop relaying for ultra-reliable communication in a factory
environment,'' in \emph{27th Annual International Symposium on Personal, Indoor, and Mobile Radio Communications (PIMRC), 2016 IEEE }.\hskip 1em plus
0.5em minus 0.4em\relax IEEE, 2016, pp. 1--6.

\bibitem{biral2016impact}
A.~Biral, H.~Huang, A.~Zanella, and M.~Zorzi, ``On the impact of transmitter
channel knowledge in energy-efficient machine-type communication,'' in
\emph{Globecom Workshops (GC Wkshps), 2016 IEEE}.\hskip 1em plus 0.5em minus
0.4em\relax IEEE, 2016, pp. 1--7.

\bibitem{durisi2016toward}
G.~Durisi, T.~Koch, and P.~Popovski, ``Toward massive, ultrareliable, and
low-latency wireless communication with short packets,'' \emph{Proceedings of
	the IEEE}, vol. 104, no.~9, pp. 1711--1726, 2016.

\bibitem{yilmaz2016ultra}
O.~Yilmaz, ``Ultra-reliable and low-latency {5G} communication,'' in
\emph{European Conference on Networks and Communications (EuCNC)}, vol. 2016,
2016.

\bibitem{sybis2016channel}
M.~Sybis, K.~Wesolowski, K.~Jayasinghe, V.~Venkatasubramanian, and
V.~Vukadinovic, ``Channel coding for ultra-reliable low-latency communication
in 5{G} systems,'' in \emph{84th Vehicular Technology Conference (VTC-Fall), 2016
	IEEE }.\hskip 1em plus 0.5em minus 0.4em\relax IEEE, 2016, pp. 1--5.

\bibitem{hu2015performance}
Y.~Hu, J.~Gross, and A.~Schmeink, ``On the performance advantage of relaying
under the finite blocklength regime,'' \emph{IEEE Communications Letters},
vol.~19, no.~5, pp. 779--782, 2015.

\bibitem{polyanskiy2010channel}
Y.~Polyanskiy, H.~V. Poor, and S.~Verd{\'u}, ``Channel coding rate in the
finite blocklength regime,'' \emph{IEEE Transactions on Information Theory},
vol.~56, no.~5, pp. 2307--2359, 2010.

\bibitem{iscan2016comparison}
O.~Iscan, D.~Lentner, and W.~Xu, ``A comparison of channel coding schemes for
{5G} short message transmission,'' in \emph{Globecom Workshops (GC Wkshps),
	2016 IEEE}.\hskip 1em plus 0.5em minus 0.4em\relax IEEE, 2016, pp. 1--6.

\bibitem{makki2015finite}
B.~Makki, T.~Svensson, and M.~Zorzi, ``Finite block-length analysis of spectrum
sharing networks,'' in International
Conference on \emph{Communications (ICC), 2015 IEEE }.\hskip 1em plus 0.5em minus 0.4em\relax IEEE, 2015, pp.
7665--7670.

\bibitem{park2012new}
J.-H. Park and D.-J. Park, ``A new power allocation method for parallel awgn
channels in the finite block length regime,'' \emph{IEEE Communications
	Letters}, vol.~16, no.~9, pp. 1392--1395, 2012.

\bibitem{devassy2014finite}
R.~Devassy, G.~Durisi, P.~Popovski, and E.~G. Strom, ``Finite-blocklength
analysis of the arq-protocol throughput over the gaussian collision
channel,'' in 6th International Symposium on \emph{Communications, Control and Signal Processing (ISCCSP),
	2014 }.\hskip 1em plus 0.5em minus 0.4em\relax
IEEE, 2014, pp. 173--177.

\bibitem{ikki2007performance}
S.~Ikki and M.~H. Ahmed, ``Performance analysis of cooperative diversity
wireless networks over nakagami-m fading channel,'' \emph{IEEE Communications
	letters}, vol.~11, no.~4, 2007.

\bibitem{zimmermann2005performance}
E.~Zimmermann, P.~Herhold, and G.~Fettweis, ``On the performance of cooperative
relaying protocols in wireless networks,'' \emph{Transactions on Emerging
	Telecommunications Technologies}, vol.~16, no.~1, pp. 5--16, 2005.

\bibitem{mansourkiaie2015cooperative}
F.~Mansourkiaie and M.~H. Ahmed, ``Cooperative routing in wireless networks: A
comprehensive survey,'' \emph{IEEE Communications Surveys \& Tutorials},
vol.~17, no.~2, pp. 604--626, 2015.

\bibitem{swamy2015cooperative}
V.~N. Swamy, S.~Suri, P.~Rigge, M.~Weiner, G.~Ranade, A.~Sahai, and
B.~Nikoli{\'c}, ``Cooperative communication for high-reliability low-latency
wireless control,'' in International
Conference on \emph{Communications (ICC), 2015 IEEE }.\hskip 1em plus 0.5em minus 0.4em\relax IEEE, 2015, pp.
4380--4386.

\bibitem{hu2016relaying}
Y.~Hu, A.~Schmeink, and J.~Gross, ``Relaying with finite blocklength: Challenge
vs. opportunity,'' in \emph{Sensor Array and Multichannel Signal Processing
	Workshop (SAM), 2016 IEEE}.\hskip 1em plus 0.5em minus 0.4em\relax IEEE,
2016, pp. 1--5.

\bibitem{du2016finite}
F.~Du, Y.~Hu, L.~Qiu, and A.~Schmeink, ``Finite blocklength performance of
multi-hop relaying networks,'' in \emph{International Symposium on Wireless
	Communication Systems (ISWCS), 2016}.\hskip 1em plus 0.5em minus 0.4em\relax
IEEE, 2016, pp. 466--470.

\bibitem{nouri2017performance}
P.~Nouri, H.~Alves, and M.~Latva-aho, ``On the performance of ultra-reliable
decode and forward relaying under the finite blocklength,'' in \emph{European
	Conference on Networks and Communications (EuCNC), 2017}.\hskip 1em plus
0.5em minus 0.4em\relax IEEE, 2017, pp. 1--5.

\bibitem{nouri2017ultra}
P.~Nouri, H.~Alves, R.~Demo~Souza, and M.~Latva-aho, ``Ultra-reliable short
message cooperative relaying protocols under nakagami-m fading,'' in
\emph{International Symposium on Wireless Communication Systems (ISWCS),
	2017}.\hskip 1em plus 0.5em minus 0.4em\relax IEEE, 2017.

\bibitem{qiao2009energy}
D.~Qiao, M.~C. Gursoy, and S.~Velipasalar, ``Energy efficiency of fixed-rate
wireless transmissions under queueing constraints and channel uncertainty,''
in \emph{Global Telecommunications Conference, 2009. GLOBECOM 2009.
	IEEE}.\hskip 1em plus 0.5em minus 0.4em\relax IEEE, 2009, pp. 1--6.

\bibitem{gursoy2009capacity}
M.~C. Gursoy, ``On the capacity and energy efficiency of training-based
transmissions over fading channels,'' \emph{IEEE Transactions on Information
	Theory}, vol.~55, no.~10, pp. 4543--4567, 2009.

\bibitem{wu2015recent}
G.~Wu, C.~Yang, S.~Li, and G.~Y. Li, ``Recent advances in energy-efficient
networks and their application in {5G} systems,'' \emph{IEEE Wireless
	Communications}, vol.~22, no.~2, pp. 145--151, 2015.

\bibitem{she2016energy}
C.~She and C.~Yang, ``Energy efficient design for tactile internet,'' in International Conference
on
\emph{Communications in China (ICCC), 2016 IEEE/CIC }.\hskip 1em plus 0.5em minus 0.4em\relax IEEE, 2016, pp. 1--6.

\bibitem{dosti2017ultra}
E.~Dosti, M.~Shehab, H.~Alves, and M.~Latva-aho, ``Ultra reliable communication
via cc-harq in finite block-length,'' in European Conference on \emph{Networks and Communications
	(EuCNC), 2017 }.\hskip 1em plus 0.5em minus 0.4em\relax
IEEE, 2017, pp. 1--5.

\bibitem{dosti2017ultraa}
E.~Dosti, U.~L. Wijewardhana, H.~Alves, and M.~Latva-aho, ``Ultra reliable
communication via optimum power allocation for type-i arq in finite
block-length,'' \emph{arXiv preprint arXiv:1701.08617}, 2017.

\bibitem{yang2014quasi}
W.~Yang, G.~Durisi, T.~Koch, and Y.~Polyanskiy, ``Quasi-static multiple-antenna
fading channels at finite blocklength,'' \emph{IEEE Transactions on
	Information Theory}, vol.~60, no.~7, pp. 4232--4265, 2014.

\bibitem{hu2016blocklength}
Y.~Hu, A.~Schmeink, and J.~Gross, ``Blocklength-limited performance of relaying
under quasi-static Rayleigh channels,'' \emph{IEEE Transactions on Wireless
	Communications}, vol.~15, no.~7, pp. 4548--4558, 2016.

\bibitem{makki2014finite}
B.~Makki, T.~Svensson, and M.~Zorzi, ``Finite block-length analysis of the
incremental redundancy harq,'' \emph{IEEE Wireless Communications Letters},
vol.~3, no.~5, pp. 529--532, 2014.

\bibitem{laneman2004cooperative}
J.~N. Laneman, D.~N. Tse, and G.~W. Wornell, ``Cooperative diversity in
wireless networks: Efficient protocols and outage behavior,'' \emph{IEEE
	Transactions on Information theory}, vol.~50, no.~12, pp. 3062--3080, 2004.

\bibitem{alves2012throughput}
H.~Alves, R.~D. Souza, G.~Fraidenraich, and M.~E. Pellenz, ``Throughput
performance of parallel and repetition coding in incremental
decode-and-forward relaying,'' \emph{Wireless Networks}, vol.~18, no.~8, pp.
881--892, 2012.

\bibitem{alves2011performance}
H.~Alves, R.~D. Souza, G.~Brante, and M.~E. Pellenz, ``Performance of type-{i}
and type-{ii} hybrid arq in decode and forward relaying,'' in 73rd \emph{Vehicular
	Technology Conference (VTC Spring), 2011 IEEE }.\hskip 1em plus 0.5em
minus 0.4em\relax IEEE, 2011, pp. 1--5.

\bibitem{athanasios2017probability}
P.~Athanasios, ``Probability, random variables, and stochastic processes,''
2017.

\bibitem{gradshteyn2014table}
I.~S. Gradshteyn and I.~M. Ryzhik, \emph{Table of integrals, series, and
	products}.\hskip 1em plus 0.5em minus 0.4em\relax Academic press, 2014.

\bibitem{waltz2006interior}
R.~A. Waltz, J.~L. Morales, J.~Nocedal, and D.~Orban, ``An interior algorithm
for nonlinear optimization that combines line search and trust region
steps,'' \emph{Mathematical programming}, vol. 107, no.~3, pp. 391--408,
2006.

\bibitem{de2011energy}
G.~G. de~Oliveira~Brante, M.~T. Kakitani, and R.~D. Souza, ``Energy efficiency
analysis of some cooperative and non-cooperative transmission schemes in
wireless sensor networks,'' \emph{IEEE Transactions on Communications},
vol.~59, no.~10, pp. 2671--2677, 2011.

\bibitem{alves2014outage}
H.~Alves, R.~D. Souza, and G.~Fraidenraich, ``Outage, throughput and energy
efficiency analysis of some half and full duplex cooperative relaying
schemes,'' \emph{Transactions on Emerging Telecommunications Technologies},
vol.~25, no.~11, pp. 1114--1125, 2014.

\bibitem{cui2005energy}
S.~Cui, A.~J. Goldsmith, and A.~Bahai, ``Energy-constrained modulation
optimization,'' \emph{IEEE transactions on wireless communications}, vol.~4,
no.~5, pp. 2349--2360, 2005.

\bibitem{tran2014achievable}
Tran, Tuyen X and Tran, Nghi H and Bahrami, Hamid Reza and Sastry, Shivakumar, ``On achievable rate and ergodic capacity of NAF multi-relay networks with CSI,'' \emph{IEEE transactions on communications}, vol.~62, no.~5, pp. 1490-1502, 2014.

\bibitem{beaulieu2006closed}
Beaulieu, Norman C and Hu, Jeremiah, ``A closed-form expression for the outage probability of decode-and-forward relaying in dissimilar Rayleigh fading channels,'' \emph{IEEE Communications Letters}, vol.~10, no.~12, 2006.

\end{thebibliography}
%

\end{backmatter}
\end{document}